\documentclass[twoside,leqno,twocolumn]{article}

\usepackage[letterpaper]{geometry}

\usepackage{ltexpprt}

\usepackage{amsfonts}
\usepackage{amsmath}
\usepackage[shortlabels]{enumitem}
\usepackage{mathtools}
\usepackage{amssymb}
\usepackage{mathrsfs}
\usepackage[numbers]{natbib}
\newtheorem{observation}{Observation}[section]
\newtheorem{conjecture}{Conjecture}[section]
\newtheorem{claim}{Claim}[section]
\newtheorem{example}{Example}[section]
\newtheorem{problem}{Open problem}[section]
\newtheorem{remark}{Remark}[section]
\newenvironment{statement}{\it \trivlist \item[\hskip \labelsep{\sc Theorem.}]}{\outerparskip 0pt\endtrivlist}

\newtheorem{definition}{Definition}

\renewcommand{\phi}{\varphi}

\newcommand{\Cc}{\mathscr{C}}
\newcommand{\Dd}{\mathscr{D}}

\newcommand{\Ff}{\mathscr{F}}

\newcommand{\N}{\mathbb{N}}

\newcommand{\strA}{\mathfrak{A}}

\newcommand{\interp}[1]{\mathsf{#1}}

\renewcommand{\mid}{~:~}
\newcommand{\ie}{i.e.\@ }

\newcommand{\limp}{\mathbin{\rightarrow}}

\newcommand{\FO}{\ensuremath{\mathrm{FO}}}
\newcommand{\MSO}{\ensuremath{\mathrm{MSO}}}

\usepackage{tikz}
\usetikzlibrary{decorations.pathreplacing}
\tikzstyle{vertex}=[circle,inner sep=1.5,minimum size =1.5mm,semithick,fill=black, draw=black]

\newcommand{\R}{\mathbb{R}}
\newcommand{\rw}{\mathrm{rw}}
\newcommand{\lrw}{\mathrm{lrw}}

\DeclareMathOperator{\Class}{\mathrm{Class}}
\DeclareMathOperator{\IC}{\mathrm{ICol}}
\DeclareMathOperator{\NC}{\mathrm{NCol}}
\newcommand{\sClass}{\mathbf{Cl}}
\newcommand{\sNC}{\mathbf{C}}

\title{Linear rankwidth meets stability}
\author{Jaroslav Ne\v set\v ril\thanks{Supported by  CE-ITI P202/12/G061 of GA\v{C}R and by the European Research Council (ERC) under the European Union's Horizon
2020 research and innovation programme (ERC Synergy Grant DYNASNET, grant agreement No 810115).}
\\[-1mm]{\small Charles University, Prague, Czech Republic}\\[-2mm] {\small \texttt{nesetril@kam.mff.cuni.cz}}
 \and Patrice Ossona de Mendez\thanks{Supported
by the European Research Council (ERC) under the European Union's Horizon
2020 research and innovation programme (ERC Synergy Grant DYNASNET, grant agreement No 810115).}
 \\[-1mm]{\small CAMS (CNRS, UMR 8557), Paris, France}\\[-2mm] {\small \texttt{pom@ehess.fr}} 
 \and Roman Rabinovich\thanks{Supported by Deutsche Forschungsgemeinschaft (DFG) --- ``Graph Classes of Bounded Shrubdepth'' Projekt number 420419861}
 \\[-1mm] {\small Technical University Berlin, 
 Germany}\\[-2mm] {\small \texttt{roman.rabinovich@tu-berlin.de}}
 \and Sebastian Siebertz
 \\[-1mm] {\small University of Bremen, Germany}\\[-2mm] {\small \texttt{siebertz@uni-bremen.de}}}
\date{}

\begin{document}
\maketitle
\begin{abstract}\small\baselineskip=9pt 
  Classes with bounded rankwidth are MSO-transductions of trees and
  classes with bounded linear rankwidth are MSO-transductions of
  paths. These results show a strong link between the properties of
  these graph classes considered from the point of view of structural
  graph theory and from the point of view of finite model theory. We
  take both views on classes with bounded linear rankwidth and prove
  structural and model theoretic properties of these classes: 1)
  Graphs with linear rankwidth at most $r$ are linearly
  \mbox{$\chi$-bounded}. Actually, they have bounded $c$-chromatic
  number, meaning that they can be colored with $f(r)$ colors, each
  color inducing a cograph. 2) Based on a Ramsey-like argument, we
  prove for every proper hereditary family $\Ff$ of graphs
  (like cographs) that there is a class with bounded rankwidth that
  does not have the property that graphs in it can be colored by a
  bounded number of colors, each inducing a subgraph in~$\Ff$.
  3) For a class $\Cc$ with bounded linear rankwidth the following
  conditions are equivalent: a) $\Cc$~is~stable, b)~$\Cc$~excludes
  some half-graph as a semi-induced subgraph, c) $\Cc$ is a
  first-order transduction of a class with bounded pathwidth.  These
  results open the perspective to study classes admitting low linear
  rankwidth covers.
  
\end{abstract}

\section{Introduction}

Hierarchical decompositions of graphs such as treewidth, pathwidth,
and treedepth \mbox{decompositions}, as well as their dense
counterparts rank-, linear rank-, and shrubdepth decompositions play
an important role in graph theory with many applications in computer
science and logic. These decompositions capture global connectivity
properties of graphs (over different connectivity functions), and
hence classes with bounded width are strongly restricted.  Via the
concept of \emph{low width covers} one can apply these decompositions
to much larger graph classes. A class of graphs $\Cc$ admits low width
covers if for every $p\in \N$ the vertices of every graph $G\in \Cc$
can be covered with a bounded number of sets $U_1,\ldots, U_N$ such
that every set $X\subseteq V(G)$ with at most $p$ elements is
contained in some~$U_i$ and such that all $U_i$ induce a subgraph of
bounded width (such a cover is a \emph{depth-$p$ width cover}). The
following characterization theorem was obtained by
\citet{nevsetvril2008grad} (originally in terms of the equivalent
concept of low treedepth colorings).
   
\vspace{-1mm}
\begin{statement}
  Let $\Cc$ be a class of graphs. The following are equivalent.\\[-6mm]
  \begin{enumerate}
  \item $\Cc$ admits low treewidth covers.\\[-6mm]
  \item $\Cc$ admits low treedepth covers.\\[-6mm]
  \item $\Cc$ has bounded expansion.
  \end{enumerate}
\end{statement}

\vspace{-1mm} Classes with bounded expansion are classes of uniformly
sparse graphs and the above decomposition by covers allows to solve
many algorithmic problems on classes with bounded expansion more
efficiently than on general graphs.  For example every property of
graphs definable in first-order logic (FO) can be tested in linear
time on graph classes with bounded expansion as shown
by~\citet{dvovrak2013testing}. This result generalizes to the more
general nowhere dense graph classes, however, it requires completely
different methods~\cite{grohe2014deciding}.

With the aim of extending the covering approach to dense graph
classes, \citet{kwon17} introduced the notion of low rankwidth covers
and low shrubdepth covers. They provided several examples of classes
that admit low rankwidth covers and showed that \mbox{$r$-powers} of
graph classes with bounded expansion admit low shrubdepth covers.  Not
much later, \citet{SBE_drops} provided a full characterization of
classes that admit low shrubdepth covers.

\vspace{-1mm}
\begin{statement}
  \label{fact:sbe}
  Let $\Cc$ be a class of graphs. The following are equivalent.\\[-6mm]
  \begin{enumerate}
  \item $\Cc$ admits low shrubdepth covers.\\[-7mm]
  \item $\Cc$ is included in a first-order transduction of a bounded
    expansion class.
  \end{enumerate}
\end{statement}

A logical transduction transforms one structure into another by
logical means. In its simplest form a transduction of a graph consists
of a coloring step and then an interpretation step, where we interpret
a new edge relation via a logical formula. For example, a transduction
can complement the edge set of a graph or connect any two vertices at
a fixed distance.  \mbox{FO-transductions} of bounded expansion
classes are named classes of \emph{structurally bounded expansion}.
Unfortunately, no efficient version of the last theorem is known.
However, if a graph is given together with a depth-$2$ shrubdepth
cover one again obtains efficient algorithms for many algorithmic
problems, in particular, one can efficiently solve the model-checking
problem of first-order logic~\cite{SBE_drops}.

\smallskip

In this work we combine graph-theoretic tools with ideas from model
theory, in particular from stability theory, to understand the
connection between rankwidth and treewidth and the connection between
linear rankwidth and pathwidth from the view of first-order logic. It
is a classical result of \citet{courcelle1992monadic} (stated
originally for cliquewidth, see also \citet[Theorem 7.47]{courcelle2012graph}) that a class~$\Cc$ of graphs has bounded rankwidth if
and only if $\Cc$ is an $\MSO$-transduction of the class of all trees
and $\Cc$ has bounded linear rankwidth if and only if $\Cc$ is an
$\MSO$-transduction of the class of all paths. Hence, the connection
between these measures from the view of monadic second-order logic is
very well understood.  In the result of Courcelle the power of MSO is
only needed to define the greatest common ancestor of two elements, and
if we represent trees by the induced tree order and paths by the
induced linear order instead, we obtain the following two characterizations
of \citet{colcombet2007combinatorial}: a class~$\Cc$ of graphs has bounded rankwidth
if and only if $\Cc$ is an $\FO$-transduction of the class of all tree
orders and~$\Cc$ has bounded linear rankwidth if and only if $\Cc$ is
an $\FO$-transduction of the class of all linear orders. This leads to
the central question about the role of order in graph classes with
bounded (linear) rankwidth and in classes that admit low (linear)
rankwidth covers. Definability of orders plays also a central role in
the area of model theory called stability theory (also known as
classification theory, cf \cite{shelah1990classification}). In a nutshell, a first-order theory is called
\emph{stable} if no formula can define arbitrarily large orders, and
this notion can naturally be adapted to infinite graph classes. These
considerations lead in particular to the following conjectures.

\begin{conjecture}\label{conjecture:rw}
  Let $\Cc$ be a class of graphs with bounded rankwidth. Then $\Cc$ is
  included in a first-order transduction of a class $\Dd$ with bounded
  treewidth if and only if $\Cc$ is stable.
\end{conjecture}

\begin{conjecture}\label{conjecture:lrw}
  Let $\Cc$ be a class of graphs with bounded linear rankwidth. Then
  $\Cc$ is included in a first-order transduction of a class $\Dd$
  with bounded pathwidth if and only if $\Cc$ is stable.
\end{conjecture}

The conjectures imply that if a class $\Cc$ is an $\FO$-transduction
of the class of all tree orders or linear orders, then either the use of
the order is witnessed by the presence of arbitrarily large
interpretable orders in the graphs in $\Cc$, or the class $\Cc$ does
not need the full power of the orders and can be obtained by an
$\FO$-transduction of a class with bounded pathwidth or treewidth. If
true, these conjectures would imply that for stable graph classes the
properties of admitting low rankwidth covers, low linear rankwidth
covers and low shrubdepth covers coincide, and are equivalent to the
property of having structurally bounded expansion.

\subsection*{Our results}

As the first main result of the paper we prove that
Conjecture~\ref{conjecture:lrw} is true. In fact we prove that the
following stronger statement is true.  A \emph{half-graph} (also known
as a \emph{ladder-} or \emph{chain-graph}) of order~$k$ is a bipartite
graph with vertex set
$\{a_1,\ldots, a_k\}\mathbin{\cup} \{b_1,\ldots, b_k\}$ and the edges
$\{a_i,b_j\}\Leftrightarrow i\leq j$ (see~Fig.~\ref{fig:HG}). We say
that a bipartite graph $H$ with parts~$A$ and~$B$ is
\emph{semi-induced} in a graph $G$, if $H$ is a subgraph of $G$ and
the edges between $A$ and $B$ in~$G$ are induced, i.e., for
\mbox{$a\in A$} and $b\in B$ we have
\mbox{$\{a,b\}\in E(H) \Leftrightarrow \{a,b\}\in E(G)$}.

\begin{figure}[ht]
\begin{center}
  \includegraphics[width=.25\textwidth]{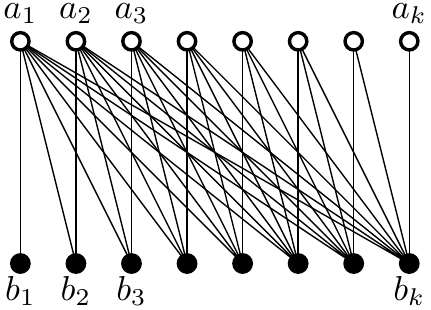}
\end{center}
\caption{The half-graph $H_k$}
\label{fig:HG}
\end{figure}

The notion of stability can be expressed in terms of half-graphs: a
formula $\phi(\bar x,\bar y)$ (where $\bar x$ and $\bar y$ denote
tuples of free variables) is {\em stable over} a class $\Cc$ of graphs
if there exists an integer $k$ such that in no graph $G\in\Cc$ one can
find tuples $\bar{a}_1,\dots,\bar{a}_k$ (of length $|\bar{x}|$) and
tuples $\bar{b}_1,\dots,\bar{b}_k$ (of length $|\bar{y}|$) with
$G\models\phi(\bar{a}_i,\bar{b}_j) \iff i\leq j$.  A class $\Cc$ is
{\em stable} if every formula is stable over $\Cc$.  In particular, if
one considers tuples of size $1$ and the formula $\phi(x,y)$
expressing adjacency of $x$ and $y$, a necessary (but generally not
sufficient) condition for a class $\mathscr C$ to be stable is that
there is an integer $k$ such that no graph in $\mathscr C$ contains
the half-graph of order $k$ as a semi-induced subgraph.

We prove the following theorem. 

\begin{theorem}\label{thm:lrw}
  Let $\Cc$ be a class of graphs with bounded linear rankwidth. Then the following are equivalent:\\[-6mm]
  \begin{enumerate}
  \item $\Cc$ is stable,\\[-6mm]
  \item  $\Cc$ excludes some semi-induced half-graph,\\[-6mm]
  \item $\Cc$ is included in a first-order transduction of a
    class~$\Dd$ with bounded pathwidth.
  \end{enumerate}
\end{theorem}

Theorem~\ref{thm:lrw} is stronger than Conjecture~\ref{conjecture:lrw}
as it does not require that every formula over $\Cc$ is stable, but it
is sufficient that the formula $E(x,y)$ is stable.  This statement is
particularly interesting in the light of the fact that graphs of
bounded linear rankwidth have bounded pathwidth if and only if they
exclude some complete bipartite graph as a subgraph as shown by
\citet{Gurski2000}.  From the point of view of first-order logic, the
obstruction for interpretability in a graph with bounded pathwidth is
not a complete bipartite graph but a half-graph. Note that the special role of
half-graphs was recently observed also by
\citet{malliaris2014regularity}, who obtained a stronger version of
Szemer\'edi's regularity lemma for graphs that exclude a semi-induced
half-graph.

As a corollary we derive that for stable graph classes~$\Cc$ the
notions of low shrubdepth covers and low linear rankwidth covers
coincide. Both of these statements are equivalent to the property that
$\Cc$ has structurally bounded expansion.

\begin{corollary}
  For a class of graphs $\Cc$ the following conditions are
  equivalent:\\[-6mm]
  \begin{enumerate}
  \item $\Cc$ is a stable class of graphs that admits low linear
    rankwidth covers.\\[-6mm]
  \item $\Cc$ admits low linear rankwidth covers and excludes some
    semi-induced half-graph.\\[-6mm]
  \item $\Cc$ admits low shrubdepth covers.\\[-6mm]
  \item $\Cc$ has structurally bounded expansion.
  \end{enumerate}
\end{corollary}

We then turn our attention to unstable classes that admit low linear
rankwidth colorings. As an order can be interpreted by first-order
formulas in such classes we must somehow handle these orders. For this
purpose we introduce a notion of \emph{embedded shrubdepth
  decompositions} that implicitly carry an order in the decomposition
tree.  We prove the following theorem.

\begin{theorem}
\label{thm:lesd}
Every class $\Cc$ with bounded linear rankwidth admits low embedded
shrubdepth covers.
\end{theorem}

As a corollary we obtain a decomposition via covers for classes with
low linear rankwidth covers.

\begin{corollary}
  A class $\Cc$ admits low linear rankwidth covers if and only if
  $\Cc$ admits low embedded shrubdepth covers.
\end{corollary}

We then return to purely graph theoretic concepts.  A class of
graphs $\Cc$ is called \emph{$\chi$-bounded} if the chromatic number
of graphs from $\Cc$ is bounded by a function of their clique number.
The concept of \mbox{$\chi$-boundedness} was introduced by
\citet{gyarfas1987problems} and has received considerable attention in
the literature. We refer to the recent survey of
\citet{scott2018survey}. As shown by \citet{dvovrak2012classes}
classes with bounded rankwidth are $\chi$-bounded.

Recall that a cograph is a graph that can be generated from the
single-vertex graph $K_1$ by joins and disjoint union. A cotree for a
cograph $G$ is a tree whose leaves are the vertices of $V(G)$ and
where the inner vertices correspond to the join and union operations
used to construct~$G$. We prove that graphs of bounded linear
rankwidth can be decomposed into parts that each induce cographs with
cotrees of bounded height.

\begin{theorem}
\label{thm:lsd1}
Let $g(r)\coloneqq (r+2)!\,2^{\binom{r}{2}}3^{r+2}$.  Every graph of
linear rankwidth $r$ can be colored with~$g(r)$ colors such that each
color class induces a cograph with cotree of height at most $r+2$.
\end{theorem}

Note that we can directly derive a weaker version of this theorem
(where we do not specify~$g(r)$) from Theorem~\ref{thm:lesd}.  As
cographs are perfect graphs, that is, graphs in which the chromatic
number of every induced subgraph equals the clique number of that
subgraph, we obtain as an immediate corollary that classes with bounded
linear rankwidth are linearly $\chi$-bounded.

\begin{corollary}
  Let $g(r)$ be as in Theorem~\ref{thm:lsd1}.  
  For every graph $G$ we have
\[\chi(G)\leq g(\mathrm{lrw}(G))\,\omega(G),\]
where $\chi(G)$ denotes the chromatic number of $G$, $\mathrm{lrw}(G)$
denotes its linear rankwidth and $\omega(G)$ its clique number.
\end{corollary}

More generally, every class with a depth-$2$ linear rankwidth cover is
linearly $\chi$-bounded.

\smallskip \citet{rw_polychi} independently announced that classes
with bounded rankwidth are polynomially $\chi$-bounded. In this case,
however, we show that the degree of the polynomial has to grow with
the rankwidth (Theorem~\ref{thm:degrw}).

In contrast, by a Ramsey argument, we show that graphs with rankwidth at most $2$ cannot be
partitioned into a bounded number of cographs (Corollary~\ref{cor:cog})
and, more generally, for every proper hereditary class $\mathscr F$
there exists an integer $r$ such that the class of graphs with
rankwidth at most~$r$ cannot be partitioned into a bounded number of
graphs in $\mathscr F$ (Corollary~\ref{cor:herw}).

\section{Preliminaries}\label{sec:preliminaries}

\noindent\textbf{Structures and logic.}  In this work we 
consider {\em{signatures}} $\Sigma$ that are finite sets of unary and
binary relation symbols, and unary function symbols.  A
\mbox{{\em{structure}}}~$\strA$ over $\Sigma$ consists of a finite
universe $V(\strA)$ and interpretations of symbols from the signature:
each unary relation symbol $U\in \Sigma$ is interpreted as a set
$U^\strA\subseteq V(\strA)$ and each binary relation symbol
$R\in \Sigma$ is interpreted as a binary relation
\mbox{$R^{\strA}\subseteq V(\strA)^2$}. Each function symbol
$f\in \Sigma$ is interpreted as a function
$f^{\strA}\colon V(\strA)\rightarrow V(\strA)$.  We omit the
superscript when the structure is clear from the context, thus
identifying each symbol with its interpretation.  If $\strA$ is a
structure and $X\subseteq V(\strA)$ then we define the
\emph{substructure} $\strA[X]$ of $\strA$ induced by $X$ in the usual
way except that for each unary function $f$ and each $a\in X$, we
define $f^{\strA[X]}(a)=f^{\strA}(a)$ if $f^{\strA}(a)\in X$ and
$f^{\strA[X]}(a)=a$, otherwise.  For a signature $\Sigma$, we consider
standard first-order logic over $\Sigma$. For a formula
$\phi(x_1,\dots,x_k)$ with $k$ free variables and a structure~$\strA$,
we define
\[
\phi(\strA)=\{(v_1,\dots,v_k)\in V(\strA)^k\mid \strA\models
\phi(v_1,\dots,v_k)\}.
\]

\smallskip\noindent\textbf{Graphs, colored graphs and trees.}
Directed graphs can be viewed as finite structures over the signature
consisting of a binary relation symbol~$E$, interpreted as the edge
relation, in the usual way. If $E$ is interpreted by a symmetric and
irreflexive relation, then the structure represents an undirected and
loopless graph. When dealing with directed graphs we denote edges by
$(u,v)$, when dealing with undirected graphs we denote them by
$\{u,v\}$. When we speak of a graph we mean an undirected graph. An
orientation of a graph $G$ is a directed graph $\vec{G}$ that for
every $\{u,v\}\in E(G)$ has exactly one of $(u,v)$ or $(v,u)$ in its
edge set.  For a finite label set $\Lambda$, by a
{\em{$\Lambda$-colored}} graph we mean a graph enriched by a unary
predicate~\(U_\lambda\)
for every $\lambda\in \Lambda$. An {\em ordered graph} is a graph that
is additionally endowed with a binary relation $<$ that is a linear
order on its vertex set.  A rooted forest is a graph $F$ without
cycles together with a unary predicate $R\subset V(F)$ selecting one
root in each connected component of $F$. A tree is a connected forest.
The {\em{depth}} of a node $x$ in a rooted forest $F$ is the number of
vertices in the unique path between~$x$ and the root of the connected
component of~$x$ in $F$. In particular, $x$ is a root of $F$ if and
only if $F$ has depth $1$ in $F$.  The depth of a forest is the
largest depth of any of its nodes. The {\em{greatest common ancestor}}
of nodes $x$ and $y$ in a rooted tree is the common ancestor of $x$
and $y$ that has the largest depth. We write $x\sqsubseteq_T y$ if $x$
is an ancestor of $y$ in a tree $T$, or simply $x \sqsubseteq y$ if
$T$ is clear from the context. The ancestor relation is also called
the \emph{tree order}.

\smallskip\noindent\textbf{Treewidth, pathwidth and treedepth.}
Treewidth is an important width parameter of graphs that was
introduced by \citet{RS-GraphMinorsII-JAlg86} as part of their graph
minors project.  Pathwidth is a more restricted width measure that was
also introduced by \citet{RS-GraphMinorsI-JCTB83}.  The notion of
treedepth was introduced by~\citet{Taxi_tdepth}.

For our purposes it will be convenient to define treewidth, pathwidth,
and treedepth in terms of intersection graphs.  Let $S_1,\ldots, S_n$
be a family of sets. The \emph{intersection graph} defined by this
family is the graph with vertex set $\{v_1,\ldots, v_n\}$ and edge set
$\{\{v_i,v_j\} : S_i\cap S_j\neq \emptyset\}$.

A \emph{chordal graph} is the intersection graph of the family of
subtrees of a tree. An {\em interval graph} is the intersection graph
of a family of intervals. A {\em trivially perfect graph} is the
intersection graph of a family of nested intervals.

The {\em treewidth} of a graph $G$ is one less than the minimum clique
number of a chordal supergraph of $G$, the {\em pathwidth} of a graph
$G$ is one less than the minimum clique number of an interval
supergraph of $G$, and the {\em treedepth} of a graph $G$ is the
minimum clique number of a trivially perfect supergraph of $G$:
\begin{align*}
  \mathrm{tw}(G)&=\min\{\omega(H)\scalebox{0.8}[1.0]{\( - \)}1: H\text{ chordal and }H\supseteq G\},\\
  \mathrm{pw}(G)&=\min\{\omega(H)\scalebox{0.8}[1.0]{\( - \)}1: H\text{ interval graph and }H\supseteq G\},\\
  \mathrm{td}(G)&=\min\{\omega(H) : H\text{ trivially perfect and }H\supseteq G\}.\\	
\end{align*}

\vspace{-5mm} A class $\Cc$ of graphs has \emph{bounded treewidth,
  bounded pathwidth, or bounded treedepth}, respectively, if there is
a bound $k \in \N$ such that every graph in $\Cc$ has treewidth,
pathwidth, or treedepth, respectively, at most $k$.

\smallskip\noindent\textbf{Rankwidth, linear rankwidth and
  shrubdepth.}  Graphs of bounded treewidth have bounded average
degree and therefore the application of treewidth is (mostly) limited
to sparse graph classes. \emph{Cliquewidth} was introduced by
\citet{courcelle1993handle} with the aim to extend hierarchical
decompositions also to dense graphs.  The notion of \emph{rankwidth}
was introduced by \citet{oum2006approximating} as an efficient
approximation to cliquewidth.  Oum and Seymour showed that cliquewidth
and rankwidth are functionally related, hence, a class $\Cc$ of graphs
has bounded cliquewidth if and only if $\Cc$ has bounded rankwidth.

For a graph $G$ and a subset $X\subseteq V(G)$ we define the
\emph{cut-rank} of $X$ in $G$, denoted~$\rho_G(X)$, as the rank of the
$|X|\times |V(G)\setminus X|$ $0$-$1$ matrix $A_X$ over the binary
field~$\mathbb{F}_2$, where the entry of $A_X$ on the $i$-th row and
$j$-th column is $1$ if and only if the $i$-th vertex in $X$ is
adjacent to the $j$-th vertex in $V(G)\setminus X$. If $X=\emptyset$
or $X=V(G)$, then we define $\rho_G(X)$ to be zero.

A \emph{subcubic} tree is a tree where every node has degree~$1$
or~$3$. A \emph{rank decomposition} of a graph~$G$ is a pair $(T,L)$,
where $T$ is a subcubic tree with at least two nodes and $L$ is a
bijection from $V(G)$ to the set of leaves of $T$.  For an edge
$e\in E(T)$, the connected components of $T-e$ induce a partition
$(X,Y)$ of the set of leaves of $T$. The \emph{width} of an edge $e$
of $(T,L)$ is $\rho_G(L^{-1}(X))$. The width of $(T,L)$ is the maximum
width over all edges of $T$. The \emph{rankwidth} $\rw(G)$ of $G$ is
the minimum width over all rank decompositions of $G$.

A \emph{cograph} is a graph that can be generated from the
single-vertex graph $K_1$ by joins, that is, we form the union and add
an edge between every two vertices of the joined graphs, and disjoint
unions. A cotree for a cograph $G$ is a tree whose leaves are the
vertices of $V(G)$ and where the inner vertices correspond to the join
and union operations used to construct $G$.

The \emph{linear rankwidth} of a graph is a linearized variant of
rankwidth, similarly as pathwidth is a linearized variant of
treewidth.  Let $G$ be an $n$-vertex graph and let $v_1,\ldots, v_n$
be an order of $V(G)$. The \emph{width} of this order is
$\max_{1\leq i\leq n-1}\rho_G(\{v_1,\ldots, v_i\})$.  The \emph{linear
  rankwidth} of~$G$, denoted $\lrw(G)$, is the minimum width over all
linear orders of $G$. If $G$ has only one vertex we define the linear
rankwidth of $G$ to be zero.  An alternative way to define the linear
rankwidth is to define a linear rank decom\-posi\-tion $(T,L)$ to be a
rank decomposition such that $T$ is a caterpillar and then define
linear rankwidth as the minimum width over all linear rank
decompositions. Recall that a caterpillar is a tree in which all the
vertices are within distance $1$ of a central path.

The analogy between treewidth and rankwidth and between pathwidth and
linear rankwidth is even more evident in view of the following result
of \citet{Gurski2000}: A class of graphs that excludes some complete
bipartite graph $K_{t,t}$ as a subgraph has bounded rankwidth if and
only if it has bounded treewidth, and it has bounded linear rankwidth
if and only if it has bounded pathwidth.

The following notion of \emph{shrubdepth} has been proposed by
\citet{ganian2012trees} as a dense analogue of treedepth.  Originally,
shrubdepth was defined using the notion of \emph{tree-models}. We
present an equivalent definition based on the notion of
\emph{connection models}, which were introduced
in~\cite{ganian2012trees} to define {\em{$m$-partite cographs}}. In
this respect, classes of bounded shrubdepth are exactly classes of
$m$-partite cographs with bounded depth.

A {\em connection model} with labels from $\Lambda$ is a rooted
labeled tree $T$ where each leaf $x$ is labeled by a label
$\lambda(x)\in \Lambda$, and each non-leaf node $v$ is labeled by a
binary relation $C(v)\subset \Lambda\times \Lambda$.  Such a model
defines a directed graph $G$ on the leaves of $T$, in which two
distinct leaves $x$ and $y$ are connected by an edge if and only if
$(\lambda(x),\lambda(y))\in C(v)$, where~$v$ is the greatest common
ancestor of~$x$ and $y$.  We say that $T$ is a \emph{connection model}
of the resulting digraph $G$. If the function $C(v)$ is symmetric for
each non-leaf node $v$, then $T$ defines an undirected graph.

A class of (di)graphs $\Cc$ has {\em bounded shrubdepth} if there is a
number $h\in \N$ and a finite set of labels~$\Lambda$ such that every
graph $G\in \Cc$ has a connection model of depth at most $h$ using
labels from $\Lambda$.

As shown by \citet{SBE_drops} a class of graphs that excludes some
complete bipartite graph $K_{t,t}$ as a subgraph has bounded
shrubdepth if and only if it has bounded treedepth.

\smallskip\noindent\textbf{Interpretations and transductions.}  In
this paper, by an {\em interpretation} of $\Sigma'$-structures in
$\Sigma$-structures we mean a transformation $\interp I$ defined by
means of formulas~$\phi_R(\bar x)$ (for $R\in\Sigma'$ of
arity~$|\bar x|$) and a formula $\nu(x)$.  For every
$\Sigma$-structure $\strA$, the $\Sigma'$-structure $\mathsf I(\strA)$
has domain~$\nu(\strA)$ and the interpretation of each relation
\mbox{$R\in\Sigma'$} is given by
\mbox{$R^{\mathsf I(\strA)}=\phi_R(\strA)\cap \nu(\strA)^{|\bar x|}$}.
A {\em monadic lift} of a $\Sigma$-structure $\strA$ is a $\Sigma^+$-expansion
$\Lambda(\strA)$ of $\strA$, where $\Sigma^+$ is the union of $\Sigma$ and a
set of unary relation symbols. 
A {\em transduction} $\mathsf T$ is the composition
$\mathsf I\circ\Lambda$ of a monadic lift and an interpretation. It is
easily checked that the composition of two transductions is again a
transduction.
If $\mathsf{T}$ is a transduction and $\Cc$ is a class of structures
we write $\mathsf{T}(\Cc)$ for $\bigcup_{G\in \Cc}\mathsf{T}(G)$.
An {\em FO-transduction} is a transduction defined from an interpretation defined using first-order formulas.
 An
{\em MSO-transduction} is defined analogously and may use MSO formulas in
the interpretation part of the transduction. Note that in this work (with exception of the next paragraph) we consider only first-order transductions.

The following characterizations of classes with bounded treewidth,
pathwidth, rankwidth, linear rankwidth, and shrubdepth show the deep
connections between these width measures and transductions.

\vspace{-2mm}
\begin{empty}
  \begin{enumerate}
  \item A class $\Cc$ of graphs has bounded treewidth (pathwidth,
    respectively) if and only if there exists an
    MSO-transduction~$\mathsf{T}$ such that the incidence graph of
    every $G\in\Cc$ is the result of applying $\mathsf{T}$ to some
    tree (path, respectively) (\cite{courcelle1992monadic} (see also
    \cite{courcelle2012graph}, Theorem 7.47)).\\[-6mm]
  \item A class $\Cc$ of graphs has bounded rankwidth (linear
    rankwidth, respectively) if and only if there exists an
    MSO-transduction~$\mathsf{T}$ such that every $G\in\Cc$ is the
    result of applying $\mathsf{T}$ to some tree (path, respectively).
    (\cite{courcelle1992monadic} (see also~\cite{courcelle2012graph},
    Theorem~7.47)).\\[-6mm]
  \item\label{it:colc} A class $\Cc$ of graphs has bounded rankwidth
    (linear rankwidth, respectively) if and only if there exists an
    FO-transduction~$\mathsf{T}$ such that every $G\in\Cc$ is the
    result of applying $\mathsf{T}$ to some tree order (linear order,
    respectively)~(\cite{colcombet2007combinatorial}).\\[-6mm]
  \item\label{it:sd} A class $\Cc$ of graphs has bounded shrubdepth if
    and only if there exists an FO-transduction~$\mathsf{T}$ and a
    height $h$ such that every $G\in\Cc$ is the result of applying
    $\mathsf{T}$ to some tree of depth at
    most~$h$~(\cite{ganian2012trees, Ganian2017}).
  \end{enumerate}
\end{empty}

To illustrate the notion of transduction we prove here the next lemma,
which will be useful in Section~\ref{sec:HG}. Recall that the
\emph{oriented chromatic number} of a graph~$G$
  is the least number $N$ such that every
orientation of~$G$ has a homomorphism to a tournament (an oriented
complete digraph) of order at most $N$ \cite{raspaud1994good, Sopena1997}.  A class of graphs has bounded
oriented chromatic number if and only if there exists a finite
tournament~$\vec T$ such that every orientation of every graph in the
class has a homomorphism to $\vec T$. For example, classes of bounded
expansion (like classes with bounded treewidth or pathwidth) have
bounded oriented chromatic number.

\vspace{-1mm}
\begin{lemma}
\label{lem:or}
Let $\mathscr C$ be a class with bounded oriented chromatic
number. Then there exists an interpretation~$\interp I$ and a number
$N$ such that for every orientation $\vec G$ of a graph $G\in\Cc$
there exists a monadic lift $G^+$ with $N$ colors such that
$\vec G=\interp I(G^+)$.

To the opposite, let $\mathscr D$ be a class of graphs with girth
growing to infinity with the order. If there exists an
interpretation~$\interp I$ and a number $N$ such that for every
orientation $\vec G$ of a graph $G\in\Cc$ there exists a monadic lift
$G^+$ with $N$ colors such that $\vec G=\interp I(G^+)$, then
$\mathcal D$ has bounded oriented chromatic number.
\end{lemma}
\vspace{-3mm}
\begin{proof}
  Assume $\Cc$ has bounded oriented chromatic number. Then there
  exists a tournament $\vec{T}$ on a finite vertex set $[N]$, such
  that every orientation $\vec{G}$ of every graph $G\in\Cc$ has a
  homomorphism to $\vec{T}$. For each $\vec{G}$ we define a monadic
  lift $G^+$ by colors in $[N]$ accordingly. Now the interpretation
  $\interp I$ orients the edges according to the orientation between
  the corr.\ vertices of $\vec{T}$ (see Fig.~\ref{fig:Tor}).
		
  Conversely, assume $\mathscr D$ is a class of graphs with girth
  growing to infinity with the order, with the property that there
  exists an interpretation~$\interp I$ and a number $N$ such that for
  every orientation $\vec G$ of a graph $G\in\Cc$ there exists a
  monadic lift $G^+$ with $N$ colors such that
  $\vec G=\interp I(G^+)$.  Let $\mathscr L(G)$ be the set of all
  these lifts of $G$.  An easy consequence of Gaifman's locality
  theorem and the Feferman-Vaught Theorem (see for instance \cite{Hodges1993}) is that for any formula
  $\phi(u,v)$ that requires that~$u$ and $v$ are adjacent there exists
  $g\in\mathbb N$ and formulas $\theta_1(x),\dots,\theta_p(x)$ with a
  single free variable and a Boolean function $F$ such that for every
  graph $G$ with girth at least~$g$ and all adjacent vertices $u,v$ of
  $G$ and every lift~$G^+$ of $G$ we have
		
  \vspace{-4mm}
  \[
  \begin{split}
    &G^+\models \phi(u,v)\\
    &\quad\iff\quad G^+\models
    F(\theta_1(u),\dots,\theta_p(u),\theta_1(v),\dots,\theta_p(v)).
  \end{split}
  \]

  \vspace{-1mm} It follows that there exists an oriented graph
  $\vec{T}$ such that for every graph $G^+\in\mathscr D$ we have a
  vertex coloring~$\gamma$ of~$G^+$ with colors in $V(\vec T)$ with
  the property that for all adjacent $u,v$, we have
  $G^+\models\phi(u,v)$ if and only if
  $(\gamma(u),\gamma(v))\in E(\vec{T})$. In other words, $\gamma$ is a
  homomorphism from the orientation of $G$ induced by the lift $G^+$
  to the oriented graph $\vec{T}$. It follows that graphs in
  $\mathscr D$ have oriented chromatic at most~$|\vec{T}|$.
\end{proof}
\begin{figure}[ht]
  \begin{center}
    \includegraphics[width=.45\textwidth]{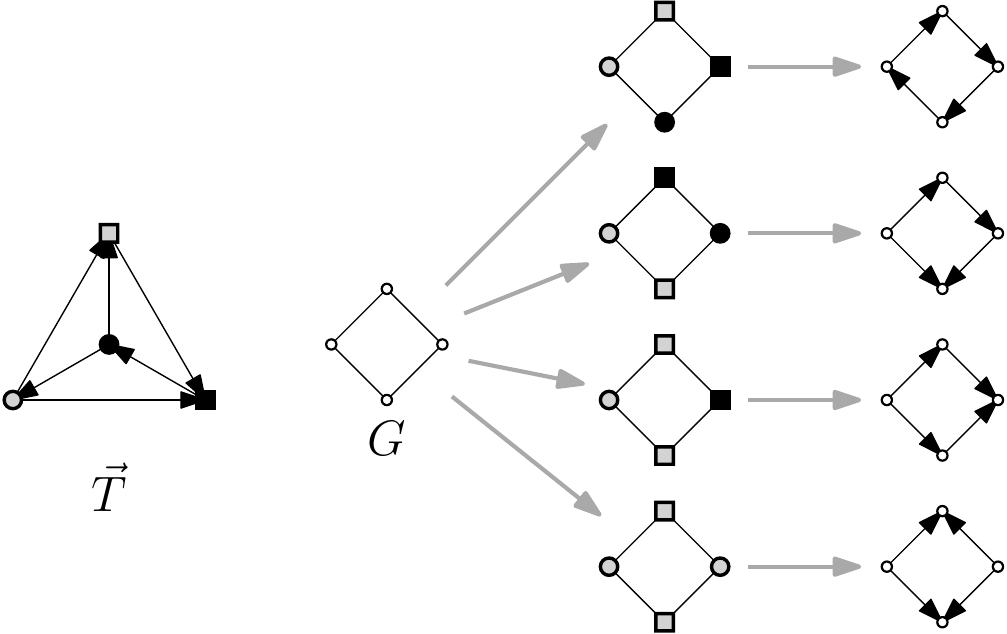}
    \vspace{-2mm}
  \end{center}
  \caption{Transductions, derived from oriented colorings, that
    produce all orientations of a graph $G$. Vertices of $G$ are
    colored by vertices of $\vec T$, and arcs between vertices of
    $\vec T$ define the orientation of edges of $G$.}
  \label{fig:Tor}
\end{figure}

\noindent\textbf{Low width covers and colorings.}
The use of the hierarchical width measures can be extended to much
more general classes via the following coloring approach. For
$p\in \N$, a \emph{$p$-treewidth coloring} is a vertex coloring
$c\colon V(G) \rightarrow C$ for some color set $C$ such that the
combination of any $p$ color classes has treewidth at most~$p$. A
class of graphs $\Cc$ is said to admit \emph{low treewidth colorings}
if it admits $p$-treewidth coloring with $N(p)$ colors for some
function $N$. This notion was introduced by \citet{devos2004excluding}
who showed that \mbox{$H$-minor} free graph classes admit low
treewidth colorings. Not much later, \citet{nevsetvril2008grad}
introduced classes of \emph{bounded expansion}, generalizing the
notion of $H$-minor free classes, and proved that these are exactly
the classes that admit low treewidth colorings and in fact \emph{low
  treedepth colorings}.  \citet{nevsetvril2011nowhere} introduced also
the more general concept of \emph{nowhere denseness} and proved that
for every nowhere dense class $\Cc$ there exists a function
$f\colon\N\times\R\rightarrow\N$ such that for every $p\in \N$ and
every $\epsilon>0$, every $n$-vertex graph $G\in\Cc$ admits a
\mbox{$p$-treedepth} coloring with $f(p,\epsilon)\cdot n^\epsilon$
colors.

The notion of low treewidth colorings extends in a natural way to any
width measure as follows. Let~$\mathrm{W}$ be a width measure. A class
of graphs admits low $\mathrm{W}$ colorings if there are two functions
$N$ and $w$ such that for every $p\in \N$, every $G\in \Cc$ can be
colored with~$N(p)$ colors such that the combination of at most~$p$
color classes has $\mathrm{W}$-width at most~$w(p)$.  Graphs with
\emph{low rankwidth colorings} were introduced and studied by
\citet{kwon17} and graphs with \emph{low shrubdepth colorings} were
studied by \citet{SBE_drops}.  \citet{SBE_drops} also introduced the
view on low width colorings as \emph{low width covers}. A class of
graphs $\Cc$ admits low $\mathrm{W}$ covers if there are two
functions~$N$ and~$w$ such that for every $p\in \N$ the vertices of
every graph $G\in \Cc$ can be covered with $N(p)$ sets
$U_1,\ldots, U_{N(p)}$ such that every set $X\subseteq V(G)$ with at
most $p$ elements is contained in some~$U_i$ and such that each
induced subgraph $G[U_i]$ has $\mathrm{W}$-width at most~$w(p)$.  It
is obvious that for every hereditary width measure $\mathrm{W}$ every
class~$\Cc$ of graphs admits low $\mathrm{W}$ colorings if and only if
it admits low~$\mathrm{W}$ covers. The view via covers is sometimes
easier to use and we will also take this view in this work.

As graphs of bounded shrubdepth are first-order transductions of
graphs of bounded treedepth and graphs of bounded expansion admit low
treedepth colorings, the following result of \citet{SBE_drops} may not
come as a surprise (even though it is not as easy to prove as it may
appear at first glance). A class of graphs admits low shrubdepth
colorings if and only if it is a first-order transduction of a bounded
expansion class.  Such classes are called classes of
\emph{structurally bounded expansion}.

The next example illustrates again the concept of simple transductions
and as a side product will provide us some examples of classes of
graphs admitting low linear rankwidth colorings.

\vspace{-2mm}
\begin{example}
  \label{ex:Lozin}
  We consider the following graph classes, introduced by
  \citet{lozin2011minimal}.  Let $n,m$ be integers. The graph
  $H_{n,m}$ has vertex set $V=\{v_{i,j}\mid (i,j)\in [n]\times [m]\}$.
  In this graph, two vertices $v_{i,j}$ and $v_{i',j'}$ with
  $i\leq i'$ are adjacent if $i'=i+1$ and $j'\leq j$. The graph
  $\widetilde{H}_{n,m}$ is obtained from $H_{n,m}$ by adding all the
  edges between vertices having the same first index (that is between
  $v_{i,j}$ and $v_{i,j'}$ for every $i\in [n]$ and all distinct
  $j,j'\in [m]$).

  First note that for fixed $a\in\mathbb N$ the classes
  \mbox{$\mathscr H_a=\{H_{a,m}: m\in\mathbb N\}$} and
  \mbox{$\widetilde{\mathscr H}_a=\{\widetilde{H}_{a,m}: m\in\mathbb
    N\}$}
  have bounded linear rank-width as they can be obtained as
  interpretations of $a$-colored linear orders: we consider the linear
  order on $\{v_{i,j}\mid (i,j)\in [a]\times [m]\}$ defined by
  $v_{i,j}<v_{i',j'}$ if $j<j'$ or $(j=j')$ and $(i<i')$. We
  color~$v_{i,j}$ by color $i$. Then the graphs in $\mathscr H_a$ are
  obtained by the interpretation stating that $x<y$ are adjacent if
  the color of $x$ is one less than the color of $y$, and if there is
  no $z$ between $x$ and $y$ with the same color as $x$. The graphs in
  $\widetilde{\mathscr H}_a$ are obtained by further adding all the
  edges between vertices with same color.
\end{example}

Following the lines of \citet[Theorem 9]{kwon17} we deduce from
Example~\ref{ex:Lozin}:

\begin{proposition}
  The class of unit interval graphs and the class of bipartite
  permutation graphs admit low linear rank-width colorings.
\end{proposition}
\vspace{-1mm}

As mentioned in the introduction, the notions of low treewidth colorings and low
treedepth colorings lead to exactly the same graph classes, namely to
bounded expansion classes. As graphs of bounded rankwidth are
transductions of trees, and hence of graphs of bounded treewidth, and
graphs of bounded shrubdepth are transductions of bounded height
trees, and hence of graphs of bounded treedepth, one may be tempted to
think that the notions of low rankwidth covers and low shrubdepth
covers also lead to the same classes of graphs, namely to classes of
structurally bounded expansion.  However, this is not true. The reason
is that classes of bounded shrubdepth have a model theoretic property
called \emph{monadic stability} that classes of bounded rankwidth do
not have.  Classes of bounded rankwidth only have the weaker model
theoretic property called \emph{monadic dependence}.  We explain these
concepts in more detail next.

\begin{figure*}[h]
  \centering \def\svgwidth{1.0\columnwidth}
  \includegraphics[width=0.95\textwidth]{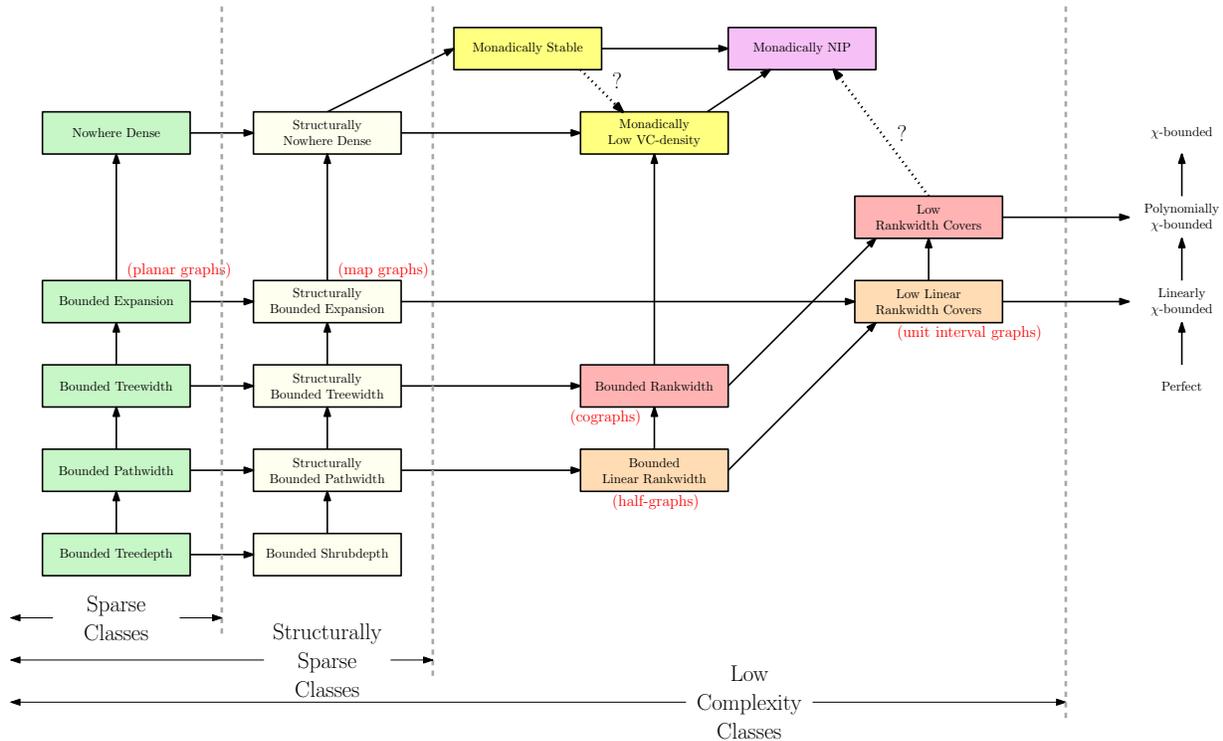}
  \caption{Inclusion map of graph classes. Some examples of classes
    are given in brackets. }
  \label{fig:Universe}
\end{figure*}

\smallskip\noindent\textbf{Stability and dependence.}  \emph{Stability
  theory}, also known as \emph{classification theory}, is a branch of
classical model theory. One of the main goals of this theory is to
classify the models of a given first-order theory according to some
simple system of cardinal invariants. We refer to the textbooks of
\citet{tent2012course}, \citet{poizat2012course},
\citet{pillay2008introduction}, and \citet{shelah1990classification}
for extensive background on stability theory.  In our context it is
most convenient to define the concepts of monadic stability and
monadic dependence in terms of the following combinatorial
configurations.

The \emph{order index} of a graph $G$ (also known as the \emph{ladder
  index} of $G$) is the largest integer $k$ such that $G$ contains a
semi-induced half-graph of order $k$. In other words, the order index
of $G$ is the largest number $k$ such that there exist
\mbox{$a_1,\ldots, a_k \in V(G)$} and $b_1,\ldots, b_k\in V(G)$ such
that $\{a_i,b_j\}\in E(G)\Leftrightarrow i\leq j$.  The
\emph{VC-dimension} of~$G$ is the largest number $k$ such that there
are $a_1,\ldots a_k\in V(G)$ and $(b_J)_{J\subseteq [k]}\in V(G)$ such
that $\{a_i,b_J\}\in E(G)\Leftrightarrow i\in J$.  A class of
graphs~$\Cc$ is called \emph{monadically stable} if for every
transduction~$\mathsf{T}$ there is a number $k$ such that the order
index of $\mathsf{T}(G)$ for every $G\in \Cc$ is bounded by $k$.  A
class of graph $\Cc$ is called \emph{monadically dependent} if for
every transduction~$\mathsf{T}$ there is a number $k$ such that the
VC-dimension of $\mathsf{T}(G)$ for every $G\in \Cc$ is bounded by
$k$.  If a class is monadically stable, then it is also monadically
dependent. We remark that in general stability and dependence are
defined in terms of first-order formulas $\phi(\bar x,\bar y)$ with
more than two free variables. However, as shown by
\citet{baldwin1985second}, for monadic stability and monadic
dependence we may restrict to formulas with only two free variables,
which fits exactly the framework of transductions.

The notion of stability is a robust notion of well behaved first-order
theories. However, it found almost no attention in graph theory until
\mbox{\citet{malliaris2014regularity}} obtained a stronger version of
Szemer\'edi's regularity lemma for graphs that exclude a semi-induced
half-graph and \citet{adler2014interpreting}, building on results of
\citet{Podewski1978}, proved that for classes~$\Cc$ of graphs that are
closed under taking subgraphs the notions of dependence, stability,
and nowhere denseness coincide.  Recently, also algorithmic
applications of stability mainly for domination problems in graphs
were found by \citet{kreutzer2018polynomial},
\citet{EickmeyerGKKPRS17}, \mbox{\citet{PilipczukST18a}}, and
\citet{FabianskiPST19}. Several important algorithmic problems (in
increasing difficulty) are whether the independent set problem, the
subgraph isomorphism problem and the first-order model-checking
problem are fixed-parameter tractable on (monadically) stable graph
classes.

Classes of bounded rankwidth are prominent examples of monadically
dependent classes. Interestingly, on classes of bounded rankwidth the
notion of stability coincides with monadic stability: by a result of
\citet{baldwin1985second} we can find arbitrarily large $1$-subdivided
complete bipartite graphs via FO-transduction in every class that is
stable but not monadically stable. However, the class of all
$1$-subdivided complete bipartite graphs does not have bounded
rankwidth and every transduction of a class of bounded rankwidth must
have again bounded rankwidth.

The class of all half-graphs has bounded (linear) rankwidth, but is
not stable. Now a simple Ramsey argument implies that classes of
bounded (linear) rankwidth do in general not admit low shrubdepth
covers, in particular, the notion of low (linear) rankwidth covers
leads to strictly more general graph classes than the notion of low
shrubdepth covers.  This leads to the central question of this work:
what is the role of order in graph classes of bounded (linear)
rankwidth and in classes that admit low (linear) rankwidth covers. In
particular, these considerations lead to
Conjecture~\ref{conjecture:rw} and Conjecture~\ref{conjecture:lrw}.
The current state of various considered classes and their inclusions
are depicted schematically in Fig.~\ref{fig:Universe}. This should be
compared with similar inclusion diagrams contained e.g.\ in
\cite{Sparsity,SBE_drops}.  This displays the rapid progress of
investigations on this boundary of finite model theory and of
structural and algorithmic graph theory.

\section{Linear rankwidth meets stability}\label{sec:lrw}
In this section we prove Theorem~\ref{thm:lrw}, that is, we prove that 
in the absence of a large semi-induced half-graph every graph
of bounded linear rankwidth is a first-order transduction 
of a graph of bounded pathwidth. 

\smallskip
Our strategy is as follows. In the following we fix a 
graph $G=(V,E)$ of linear rankwidth at most~$r$ witnessed by a linear order $<$ on its vertex set
that satisfies $\max_{v\in V}\rho_G(\{u: u<v\})
\leq r$.  

Our first aim is to encode~$G$
with additional colors only with reference to $<$ and not
to the edges of $G$. For every vertex $u$ we define an interval $I_u$ starting at $u$ and ending at some larger $\tau(u)$. For vertices $v>u$ we will encode whether $\{u,v\}\in E(G)$ either 
directly by colors of $u$ and $v$ if $u<v<\tau(u)$, or we delegate the question to some vertex set $F_0$ of vertices 
smaller than $u$ that represent 
the neighborhood of~$u$ after $\tau(u)$. The main 
technical challenge is to avoid long chains of 
delegations that cannot be resolved by first-order transductions. 

In a second step we will use the additional assumption that
we exclude a half-graph to get rid of the reference to the order
and encode the required information in the intersection
graph of the intervals $I_v$, which has bounded
pathwidth. 

A similar construction is presented by 
\citet{kwon2014graphs}, who show that a class 
of graphs has linear rankwidth $k$ if and only 
if it is a pivot-minor of a graph of pathwidth $k+1$.
Our representation could be derived from the
work of \citet{kwon2014graphs}, however, 
we prefer to give our own presentation that
is tailored to decoding by first-order transductions.

\subsection{Notation.}

For sets $M,N\subseteq V(G)$ we define $M\oplus N$ as 
the symmetric difference of $M$ and $N$, that is, 
$v\in M\oplus N$ if and only if $v\in M\cup N$ but 
$v\notin M\cap N$. 

For $t\in V$, we define $V^{>t}\coloneqq 
\{v : v>t\}$, $V^{<t}\coloneqq 
\{v : v<t\}$ and $V^{\leq t}\coloneqq \{v : v\leq t\}$. 
For $v\in V$ we denote by~$N(v)$ the neighborhood 
of $v\in G$ (where $v$ not included). We let $N^{<t}(v)\coloneqq N(v)\cap V^{<t}$ and define similarly $N^{>t}$ and $N^{\leq t}$.
For $M\subseteq V(G)$ we define $N_\oplus(M)\coloneqq
\bigoplus_{v\in M} N(v)$ and $N_\oplus^{>t}(M)\coloneqq
N_\oplus(M)\cap V^{>t}$. 

\begin{remark}\label{rem:larger-neighborhoods}
If $t<t'$, then $N_\oplus^{>t}(M)=N_\oplus^{>t}(N)$ implies 
$N_\oplus^{>t'}(M)=N_\oplus^{>t'}(N)$. 
\end{remark}

For $t\in V$ the closure of $\{N^{>t}(v) : v\leq t\}$ under
$\oplus$ is a vector space over $\oplus$ and scalar multiplication
with~$0$ and $1$, where $0\cdot M=\emptyset$ and 
$1\cdot M=M$. 

Slightly abusing notation, for $t\in V$, we call an inclusion-wise 
minimal subset $B\subseteq V_{\leq t}$ a \emph{neighbor basis 
for $V^{>t}$} if for every
$v\leq t$ there exists $B'\subseteq B$ such that 
$N^{>t}(v)=N_\oplus^{>t}(B')$. We say that~$N^{>t}(v)$ is a \emph{linear combination} of the
neighborhoods~$N^{>t}(w)$ for $w\in B'$.
In other words, $B$ is a neighbor basis for $V^{>t}$ if 
$\{N^{>t}(v)\mid v\in B\}$ forms a basis for the space spanned by $\{N^{>t}(v)\mid v\leq t\}$. 

\smallskip
The following is 
immediate by the definition of linear rankwidth. 

\begin{remark}
As $G$ has linear rankwidth at most $r$, for every $t\in V$ 
there exists a neighbor basis for $V^{>t}$ of order at most $r$. Note that
$\emptyset$ is a neighbor basis for $V^{>\max V}$. 
\end{remark}

\subsection{Activity intervals and active basis.}

Consider a vertex $v$ and its neighborhoods $N^{>t}(v)$ for
$t\geq v$. Of course, we have
$N^{>v}(v) \supseteq N^{>t}(v)$. For
$t\in V$ the \emph{active basis $B_t$ at~$t$} is the
lexicographically least neighborhood basis of $V^{>t}$. If  we have $v \notin B_t$, then $v\notin B_{t'}$ for all
$t'>t$ by Remark~\ref{rem:larger-neighborhoods}. In other words, from some point~$t$ onward, the neighborhood of $v$
does not contribute to lexicographically least bases because
$N^{>t}(v)$ is a linear combination of the neighborhoods of vertices
smaller than $v$. 

\begin{remark}\label{rem:Bnt}
If $t\in V$, then $v\in B_t$ if and only if there does not 
exist $B\subseteq V^{<v}$ such that $N^{>t}(v)=N_\oplus^{>t}(B)$. 
\end{remark}

To each $v\in V$ we associate its {\em activity
  interval}~$I_{v}$ defined as the interval $[v,\tau(v)] $
starting at $v$ and ending at the minimum vertex $\tau(v)\geq v$
such that for every $t$ with $v\leq t<\tau(v)$ we have $v\in B_t$. 
According to Remark~\ref{rem:Bnt}, $\tau(v)$ is the minimum vertex
$\tau(v)\geq v$ 
such that there exists $B\subseteq V^{<v}$ with
$N^{>\tau(v)}(v)=N_\oplus^{>\tau(v)}(B)$.
Note that $\tau(v)$
is well defined as $N^{>\max V}(v)=N_\oplus^{>\max V}(\emptyset)= \emptyset$.

\begin{remark}\label{obs:vBt}
If $t\in V$, then for $v\leq t$ we have 
$v\in B_t$ if and only if $t<\tau(v)$. 
\end{remark}

We extend the definition of activity intervals to all sets $M\subseteq V(G)$ by
\begin{equation}
	I_M\coloneqq \bigcap_{v\in M}I_v\quad \text{ and } \quad \tau(M)=\min_{v\in M}\tau(v).
\end{equation}

Note that either $I_M=\emptyset$ or $I_M=[\max M,\tau(M)]$.
We call a set $M$ {\em active} if $|I_M|>1$, that is, if
\mbox{$\max M<\tau(M)$}. 
We call a vertex $v$ \emph{active} if the single\-ton set $\{v\}$
is active. 

\smallskip
For every $v\in V$, as $v\notin B_{\tau(v)}$, there 
exists a unique $F_0(v)\subseteq B_{\tau(v)}$ with 
\begin{equation}
\label{eq:F0}
N^{>\tau(v)}(v) = N_\oplus^{>\tau(v)}(F_0(v)).
\end{equation}	

Note that if $F_0(v)\neq \emptyset$, then we have 
\begin{equation}
\label{eq:interval}
\max F_0(v)<v\leq \tau(v)<\tau(F_0(v))	,
\end{equation}

hence, in this case, the set $F_0(v)$ is active.

\begin{remark}\label{rem:active-F}
Assume that $M$ is an active set and let $v\in M$. 
\begin{enumerate}
\item If 
$\tau(v)> \tau(M)$, then $v\in B_{\tau(M)}$. 
\item If $\tau(v)= \tau(M)$, then $F_0(v)
\subseteq B_{\tau(M)}$. 
\end{enumerate}
\end{remark}

Let us pause for a moment. We have established the
key notation and now we can describe the idea of
encoding $G$ in a colored linear order more precisely. As  
explained above, 
for vertices $v>u$ we will encode whether $\{u,v\}\in E(G)$ either 
directly by colors of~$u$ and~$v$ if $u<v<\tau(u)$, or we delegate 
the question to the set~$F_0(u)$ that represents 
the neighborhood of~$u$ after $\tau(u)$. This process of delegating 
is justified  by \eqref{eq:F0}. The problem may arise
that the vertices of $F_0$ themselves do again not directly 
encode whether they are adjacent to $v$, but delegate this
information again. We show that by our choice of representing
neighborhoods by minimal bases this referencing process must
stop after a bounded number of steps.

\subsection{The F-tree.}

We define a mapping $F$ extending~$F_0$, that will define a rooted tree on the set  $Z$ consisting of all active sets, all singleton sets $\{v\}$
for $v\in V(G)$, and $\emptyset$ (which will be the root of the tree and the unique fixed point of $F$). Before we define $F$ we make 
one more observation. 

\begin{lemma}
\label{lem:tau}
	Let $u,v\in V(G)$ be active. If $\tau(u)=\tau(v)$, then $u=v$.
\end{lemma}
\begin{proof}
  Let $t=\tau(u)=\tau(v)$ and let $t'$ be the predecessor of $t$ in the linear order.
  By definition of $F_0$ we have
	$N^{>t}(u)=N^{>t}(F_0(u))$ and $N^{>t}(v)=N^{>t}(F_0(v))$.
	We have $N^{>t'}(u)\neq N^{>t'}(F_0(u))$ as otherwise
	$\tau(u)\leq t'$. As $N^{>t'}(u)\oplus N^{t}(u)\subseteq \{t\}$ and $N^{>t'}(F_0(u))\oplus N^{t}(F_0(u))\subseteq \{t\}$, we have
	$N^{>t'}(F_0(u))=N^{>t'}(u)\oplus\{t\}$.
	Similarly, we have $N^{>t'}(F_0(v))=N^{>t'}(v)\oplus\{t\}$. Assume without loss of generality that $u<v$. Then
	$N^{>t'}(v)=N^{>t'}(\{u\})\oplus N^{>t'}(F_0(u))\oplus N^{>t'}(F_0(v))$. 
	As $\max (\{u\}\cup F_0(u)\cup F_0(v))<v$ we deduce that $\tau(v)\leq t'$, contradicting $\tau(v)=t$.
	\hfill
\end{proof}

\begin{corollary}
For each active set $M\subseteq V(G)$ there exists exactly one $v\in M$
with $\tau(v)=\tau(M)$. 
\end{corollary}

The mapping $F\colon Z\rightarrow Z$ is defined now as 
\begin{equation}
	F(M)=\begin{cases}
		\emptyset&\text{if }M=\emptyset,\\
\\
M\oplus \{v\}\oplus F_0(v) &\parbox{.2\textwidth}{for the 
unique $v\in M$\linebreak
with $\tau(v)=\tau(M)$\linebreak otherwise.}
	\end{cases}
\end{equation}
\medskip

The following lemma shows for every active set $M$, either 
$F(M)=\emptyset$ or $F(M)$
is active, and thus $F(M)\in Z$ and $F$ is well defined. 
Furthermore, the lemma shows that 
$I_{F(M)}\supset I_M$. 

\begin{lemma}
\label{claim:tauinc}
Let $M\in Z$. Then $F(M)\subseteq B_{\tau(M)}$ and furthermore, 
either $F(M)=\emptyset$, or 
$\max F(M)\leq \max M<\tau(M)<\tau(F(M))$ and 
hence $F(M)$ is active. 
\end{lemma}
\begin{proof}
The statement is obvious if $M=\emptyset$. 
For \mbox{$M=\{v\}$}, the statement is immediate from the 
definition of $F_0(v)$ and \eqref{eq:interval}. 
For all other $M\in Z$, according to Remark~\ref{rem:active-F}
we have for each $v\in M$
either $v\in B_{\tau(M)}$ if $\tau(v)> \tau(M)$, 
or $F_0(v)\subseteq B_{\tau(M)}$ if $\tau(v)=\tau(M)$. 
This implies $F(M)\subseteq B_{\tau(M)}$. 
Finally, if $F(M)\neq \emptyset$, then 
$\max F(M)\leq \max M<\tau(M)<\tau(F(M))$ follows
from the fact that these inequalities hold for all 
$v\in M$ with $\tau(v)>\tau(M)$ and for $F_0(v)$
for the unique $v\in M$ with $\tau(v)=\tau(M)$ according to \eqref{eq:interval}. \hfill
\end{proof}

\smallskip
The mapping $F$ guides the process of iterative referencing
and ensures that, for an active set~$M$, if $t\geq\tau(M)$, then the 
set $N_\oplus^{>t}(M)$ can be rewritten as 
$N_\oplus^{>t}(F(M))$. 
This property is stated in the next lemma.

\begin{lemma}
\label{lem:F}
Let  $M\in Z\setminus\{\emptyset\}$ and let $w\in V(G)$. 
If $w>\tau(M)$, then 
\[
    w\in N_\oplus(M) \Leftrightarrow w\in N_\oplus(F(M)). 
\]
\end{lemma}
\begin{proof}
If $M=\{v\}$ for $v\in V(G)$, then this follows from~\eqref{eq:F0}.
Otherwise, $M$ is an active set.
Let $t=\tau(M)$ and let $v\in M$ be the unique element
with $\tau(v)=t$. 
Then we have $N_\oplus^{>t}(F(v))=
N_\oplus^{>t}(v)$, and hence $N_\oplus^{>t}(F(M))
= \bigoplus_{w\in F(M)}N^{>t}(w)=\bigoplus_{w\in M\setminus\{v\}}
N^{>t}(w)\mathbin{\oplus} N^{>t}_\oplus F_0(v)=\bigoplus_{w\in M\setminus\{v\}} N^{>t}(w)\mathbin{\oplus} N^{>t}(v)=
N_\oplus^{>t}(M)$. Thus for every $w>\tau(M)$ we have 
$w\in N_\oplus(M) \Leftrightarrow w\in N_\oplus(F(M))$.
\hfill
\end{proof}

This lemma can be applied repeatedly to $M, F(M),$ etc.\ until 
$F^k(M)=\emptyset$, or until for some given $w\in V(G)$ we
have $\tau(F^k(M))\geq w$. This justifies to introduce, for 
distinct vertices $u$ and $v$, the value
\begin{equation}
	\xi(u,v):=\min\{k\mid v\in I_{F^k(u)}\text{ or }F^k(u)=\emptyset\}.
\end{equation}

Until Section~\ref{sec:HG} we will consider only the case where $u<v$, which will allow us to give an alternative expression for $\xi(u,v)$.
\begin{remark}
If $u<v$ then
\begin{equation*}
	\xi(u,v)=\min\{k\mid \tau(F^k(u))\geq v\text{ or }F^k(u)=\emptyset\}.
\end{equation*}	
\end{remark}

\smallskip
As a direct consequence of the previous lemma we have
\begin{corollary}
\label{cor:fund}
For distinct $u,v\in V(G)$ we have
\[
\{u,v\}\in E(G) \Longleftrightarrow
\begin{cases}
	v\in N_\oplus(F^{\xi(u,v)}(u))&\text{if }u<v,\\
	u\in N_\oplus(F^{\xi(v,u)}(v))&\text{if }u>v.
\end{cases}
\]
\end{corollary}
\begin{proof}
As the two cases are symmetric, we can assume $u<v$. 
We prove the statement by induction on $k=\xi(u,v)$.
If $k=0$, then the statement is 
$\{u,v\}\in E(G)\Leftrightarrow v\in N_\oplus(u)$, which 
trivially holds. 
Assume $\xi(u,v)=k\geq 1$. By Claim~\ref{claim:tauinc} we have 
	$v>\tau(F^{k-1}(\{u\}))>\tau(F^{k-2}(\{u\}))>\dots>\tau(u)$.
	Moreover, $u,F(\{u\}),\dots, F^{k-1}(\{u\})\in 
	Z\setminus\{\emptyset\}$. 
	Hence  by Lemma~\ref{lem:F} we have
\begin{align*}
\{u,v\}\in E(G) & \Leftrightarrow v\in N_\oplus(u) \\
& \Leftrightarrow v\in N_\oplus(F(\{u\}))\\
& \Leftrightarrow v\in N_\oplus(F^2(\{u\})) \Leftrightarrow\ldots\\
& \Leftrightarrow v\in N_\oplus(F^k(\{u\})).
\end{align*}
\hfill
\end{proof}
The monotonicity property of $F$ (\ie the property $\tau(F(M))>\tau(M)$ if $F(M)\neq \emptyset$) implies that $F$ defines a rooted tree, \emph{the $F$-tree}, with vertex set~$Z$, root $\emptyset$ and edges $\{M, F(M)\}$. Here the monotonicity guarantees that the graph is acyclic and it is connected because~$\emptyset$ is the only fixed point of~$F$. The following lemma shows that the $F$-tree actually has bounded height. Recall that~$r$ denotes the linear rankwidth of~$G$. 

\begin{lemma}
	For every $M\in Z$ we have $F^{r+1}(M)=\emptyset$.
\end{lemma}
\begin{proof}
If $M=\emptyset$, the statement is obvious, so assume $M\neq \emptyset$. It is sufficient 
to prove that for every active set~$M$ we have 
$F^{r}(M)=\emptyset$, as this implies \mbox{$F^{r+1}(\{v\})=\emptyset$}
also for all $v\in V(G)$.
Let $M$ be an active set and let $t \in I_M$. Then every $v\in M$ is in $B_t$, so $M\subseteq B_t$. 

Assume $i\geq 1$ is such that $F^i(M)\neq \emptyset$. 
As $\max F(M)\leq \max M$ and $\tau(F(M))>\tau(M)$ 
by Claim~\ref{claim:tauinc}, we get
\[
\begin{split}
\max F^i(M)&\leq \max M\leq t<\tau(M)\\
&\leq \tau(F^{i-1}(M))<\tau(F^i(M)).
\end{split}
\]
By Observation~\ref{obs:vBt} we have $v\in B_t\Leftrightarrow t<\tau(v)$. 
As $\tau(F^i(M))=\min_{v\in F^i(M)}\tau(v)$, we 
have $F^i(M)\subseteq B_t$. Hence, 
considering the sequence $M, F(M), \dots, F^i(M)$, each 
iteration of~$F$ removes the unique element with minimum 
$\tau$ value. It follows that the union of the sets 
has cardinality at least $i+1$. As $|B_t|\leq r$, we have $i<r$ and
hence $F^r(M)=\emptyset$.\hfill
\end{proof}

For distinct vertices $u,v$, let $u\wedge v$ denote the greatest common ancestor of $u$ and $v$ in the $F$-tree, \ie the first common vertex on the paths to the root.
Then there exist $\ell_u$ and~$\ell_v$ such that
$u\wedge v=F^{\ell_u}(u)=F^{\ell_v}(v)$, hence both $u$ and $v$ belong 
to $I_{u\wedge v}$. Thus we have $\tau(u\wedge v)>u$ and $\tau(u\wedge v)>v$. In other words, we have
 $\xi(u,v)\leq \ell_u$ and $\xi(v,u)\leq \ell_v$.

\subsection{The activity interval graph}

Let $H$ be the intersection graph of the intervals $I_v$ 
for $v\in V(G)$. Note that we may identify $V(H)$ with 
$V(G)$ as $\min I_v=v$ for all $v\ V(G)$.

\begin{lemma}
The intersection graph $H$ of the intervals~$I_u$ has pathwidth
at most $r+1$, \ie at most $r+2$ intervals intersect in each point. 
\end{lemma}
\begin{proof}
Consider any vertex $t$ with $t\in I_u$ for some~$u$. The case $u\in B_t$ 
gives a maximum of $r$ intervals intersecting in~$t$. Otherwise 
$t=\tau(u)$, which gives at most two possibilities for $u$: 
either $u$ is inactive (and $u=t$), or $u$ is active
(and $u$ is uniquely determined, according to Lemma~\ref{lem:tau}).
Thus at most $r+2$ intervals intersect at point~$t$.\hfill
\end{proof}

As mentioned in the proof of the above lemma, every clique of $H$ contains at most one inactive vertex.
It follows that there is a coloring $\gamma\colon V(G)\rightarrow[r+2]$ with the following properties:
\begin{enumerate}[(1)]
	\item for every $u\in V(G)$ we have $\gamma(u)=r+2$ if and only if $u$ is inactive;
	\item for all distinct $u,v\in V(G)$ we have
\begin{equation}
I_u\cap I_v\neq\emptyset\quad\Longrightarrow\quad\gamma(u)\neq\gamma(v).
\end{equation}
\end{enumerate} 

We extend this coloring to sets as follows:
for $M\subseteq V(G)$ we let
\begin{equation}
\Gamma(M)\coloneqq \{\gamma(v)\mid v\in M\}.	
\end{equation}

This coloring allows to define, for each $v\in V(G)$ 
\begin{align*}
	\Class(v)&\coloneqq \big(\gamma(v),\Gamma(F(v)),\dots,\Gamma(F^r(v))\big),\\
	\NC(v)&\coloneqq \{\gamma(u)\mid u\in N(v)\text{ and }v\in I_u\}\\
	\IC(v)&\coloneqq \{\gamma(u)\mid v\in I_u\}\\
\end{align*}
Note that all $u$ with $v\in I_u$ define a 
clique of $H$ (because all $I_u$ contain~$v$) and hence have distinct $\gamma$-colors.

\begin{lemma}
\label{cl:gamma}
	Let $v\in V(G)$. Every $u\in B_v$ can be defined as the 
	maximum vertex $x\leq v$ with $\gamma(x)=\gamma(u)$. 
\end{lemma}
\begin{proof}
	By assumption we have $u\leq v$. Assume towards a contradiction that there exists $x\in V(G)$ with \mbox{$u<x\leq v$} and $\gamma(x)=\gamma(u)$. As $u\in B_v$ we have $\tau(u)>v$, hence $x\in I_u$. It follows that 
	$I_x\cap I_u\neq \emptyset$, in contradiction to $\gamma(x)=\gamma(u)$.\hfill
\end{proof}

Towards the aim of bounding the number of graphs of linear 
rankwidth at most $r$, we give a bound on the number of 
colors that can appear. 

\begin{lemma}
\label{cl:f}
Let $f(r)\coloneqq 3(r+2)!\,2^{\binom{r+1}{2}}$. 
The number of pairs $(\Class(v),\NC(v))$ for $v\in V(G)$ can be bounded by $f(r)$.
\end{lemma}
\begin{proof}
		Let $v\in V(G)$.
	From the fact that $\gamma(v)=r+2$ if and only if $v$ is inactive, that images by $F$ only contain active vertices, as well as from Claim~\ref{claim:tauinc} we deduce:
\begin{itemize}
	\item 	If $\gamma(v)=r+2$, then there exists a linear order on $[r+1]$ colors such that for $1\leq i\leq r$, the set 
	$\Gamma(F^i(v))$ is a subset of the first $r+1-i$ colors of~$[r+1]$.
	\item  If $\gamma(v)\leq r+1$, then there exists a linear order on $[r+1]\setminus\{\gamma(v)\}$ such that for $1\leq i\leq r$, the set $\Gamma(F^i(v))$ is a subset of the first $r-i$ colors of $[r]$.
\end{itemize}

Thus the number of distinct $\Class(v)$ for $v\in V(G)$ is bounded by
\[
(r+1)!\,2^{r}2^{r-1}\dots 2+(r+1)r!\,2^{r-1}\dots 2=3(r+1)!\,2^{\binom{r}{2}}.
\]
Furthermore, the number of distinct $\NC(v)$ for $v\in V(G)$ is at most $(r+2)2^{r+1}$.\hfill
\end{proof}
\begin{lemma}
	Let $f'(r)\coloneqq (r+2)!\,2^{\binom{r}{2}}3^{r+2}$. The number of triples $(\Class(v),\NC(v),\IC(v))$ for $v\in V(G)$  can be bounded by $f'(r)$.
\end{lemma}
\begin{proof}
	In Lemma~\ref{cl:f} we have shown that the number of distinct $\Class(v)$ for $v\in V(G)$ is bounded by $3(r+1)!\,2^{\binom{r}{2}}$.
	The number of pairs $(\NC(v),\IC(v))$ is at most $(r+2)3^{r+1}$ (for each color $a$ in $[r+1]$ either
	$a\notin\IC(v)$ or $a\in\IC(v)\setminus \NC(v)$ or $a\in \NC(v)$).
	
	\hfill
\end{proof}

\subsection{Encoding G in the linear order.}\label{subsec:encoding}

We first make use of Corollary~\ref{cor:fund} to encode $G$ 
by a first-order formula using only the newly added colors
and the order~$<$ on $V(G)$. More precisely, 
let $\mathcal{L}$ be the structure over signature 
$\Lambda\mathbin{\cup} \{<\}$, where $\Lambda$ is the
set of all colors of the form $(\Class(v), 
\NC(v),\IC(v))$, with the same elements as $G$ and $<$ interpreted as
in~$G$. Every element $v$ of $\mathcal{L}$ is equipped with 
the color $(\Class(v), \NC(v),\IC(v))$. The following lemma gives a new
proof of a result of \cite{colcombet2007combinatorial}.

\begin{lemma}
There exists an $\exists\forall$-first-order formula $\phi(x,y)$ over the 
vocabulary $\Lambda\mathbin{\cup} \{<\}$ such that
for all $u,v\in V(G)$ we have 
\[\mathcal{L}\models\phi(u,v)\Longleftrightarrow \{u,v\}\in E(G).\]
\end{lemma}
\begin{proof}
By symmetry, we can assume that $u<v$. According to Corollary~\ref{cor:fund} for distinct $u,v\in V(G)$ we have
\[
\{u,v\}\in E(G) \Longleftrightarrow
\begin{cases}
	v\in N_\oplus(F^{\xi(u,v)}(u))&\text{if }u<v\\
	u\in N_\oplus(F^{\xi(v,u)}(v))&\text{if }u>v.
\end{cases}
\]

Note that we can extract any color from $\Lambda$, \ie we can define $\gamma(x) \in \Gamma(F^i(y))$ and $\gamma(x) \in \IC(y)$. For example, $\gamma(x) \in \Gamma(F^i(y))$ is a big disjunction over all possible colorings 
$\Lambda(x) = (\Class(x),NC(x),\IC(x))$ and $\Lambda(y) = (\Class(y),NC(y),\IC(y))$ satisfying that $\Class(x)$ has in 
its first component an element from the $i$th component of $\Class(y)$.

We first define formulas $\psi^i(x,y)$ such that for 
all $u,v\in V(G)$
\[\mathcal{G}\models \psi^i(u,v)\Leftrightarrow v\in F^i(u).\]

Let $C=\Gamma(F^i(u))$. According to Lemma~\ref{cl:gamma}, 
for $a\in C$, the element of $F^i(u)\subseteq B_u$ 
with color~$a$ is the maximal element 
$w<u$ such that $\gamma(w)=a$. The formula can express
that $y< x$ is maximal with $\gamma(y)=a$ by 
$(y<x)\wedge (\gamma(y)=a) \wedge \forall z\, ((z>y)\wedge $ $(z<x) \rightarrow \gamma(z)\neq a)$. 
Here, for convenience, 
we use $\gamma(z)=a$ as an atom. Note that $\psi^i(x,y)$ is a 
$\forall$-formula. 

\smallskip
We now define formulas $\alpha^k(x,y)$ such that for all 
$u,v\in V(G)$
with $u<v$ we have 
\[\mathcal{G}\models\alpha^k(u,v)\Leftrightarrow k=\xi(u,v).\]

Observe that $v\in I_{F^k(u)}$ if and only if for every $x\in F^k(u)$ we have $x\leq v$, $a\in \IC(v)$ (i.e. there exists some $y$ with $\gamma(y)=a$ and $v\in I_y$) and there exists no~$z$ with $x<z\leq v$ with $\gamma(z)=a$ (hence $\min I_y\leq x$, which implies that $I_y$ and $I_x$ intersects thus $x=y$ as $\gamma(x)=\gamma(y)$).
We restrict ourselves to the case $u<v$ and obtain 
\begin{equation*}
\begin{split}
&u<v \land v\in I_{F^k(u)}\iff\\ 
&\qquad \quad u<v  \land{} \Gamma(F^k(u))\subseteq \IC(v)\\ 
&\qquad \land \forall x\, (x\in F^k(u) \limp x \leq v \land \gamma(x)\notin \IC(v))\,.
\end{split}
\end{equation*}

Then $\xi(u,v)$ for $u<v$ is the minimum integer~$k$ such that $v\in I_{F^k(u)}$ or $F^k(u)=\emptyset$, and
  this is easy to state as a $\forall$-formula. 
\smallskip
Finally, if we have determined $\xi(u,v)$, with the help of the
formulas $\psi^i$ we can determine whether $\{u,v\}\in E(G)$
as in the proof of Corollary~\ref{cor:fund} by existentially quantifying
the elements of $F(u), F^2(u),\ldots, F^{\xi(u,v)}(u)$ and 
expressing whether $v\in N_\oplus(F^{\xi(u,v)}(u))$. 
Indeed, for every $x\in F^{\xi(u,v)}(u)$ we have $v\in I_{F^{\xi(u,v)}(u)}\subseteq I_x$, hence the adjacency of $x$ and $y$ is encoded in $\NC(v)$.

This information can hence
be retrieved by an $\exists\forall$-formula, as claimed. 
\end{proof}

As a corollary we conclude an upper bound on the number of 
graphs of bounded linear rankwidth. The number of distinct 
values of $(\Class(v),\NC(v),\IC(v))$ is at most $f'(r)$. 
Hence we have the following upper bound.

\begin{corollary}
	Unlabeled graphs with linear rankwidth at most $r$ can be encoded using at most 
	$\binom{r}{2}+r\log_2 r+\log_2(3/e)r+O(\log_2 r)$ bits per vertex.
	Precisely, the number of unlabelled graphs of order $n$ with linear rankwidth at most $r$ is at most $\left[(r+2)!\,2^{\binom{r}{2}}3^{r+2}\right]^n$.
\end{corollary}

\subsection{Partition into Cographs.}
\label{sec:cograph}
The {\em c-chromatic number} $c(G)$ of a graph $G$, introduced by \cite{Gimbel2002705}, is the minimum order of a partition of $V(G)$ where each part induces a cograph. It is well known that cographs
are perfect graphs, hence $\chi(G)=\omega(G)$ for every cograph~$G$. 
Hence a partition of a graph $G$ into a bounded number of cographs
immediately gives a linear dependence between $\chi(G)$ and 
$\omega(G)$. 

\begin{theorem}
\label{thm:cog}
Let $f(r)=3(r+2)!\,2^{\binom{r+1}{2}}$.
For every graph $G$, we have $c(G)\leq f(\lrw(G))$ and hence 
\begin{equation}
	\chi(G)\leq f(\lrw(G))\,\omega(G).
\end{equation}
\end{theorem}
\begin{proof}
	Let $u\sim v$ hold if and only if $\Class(u)=\Class(v)$ and $\NC(u)=\NC(v)$. 
As proved in Lemma~\ref{cl:f} there are at most $f(r)$ equivalence classes for the relation $\sim$.
		
	Let $X$ be an equivalence class for $\sim$, and let $u,v$ be distinct elements in $X$.
	Let $k=\xi(u,v)$ and let $\ell=\xi(v,u)$. 

	If $F^k(u)=\emptyset$, then $F^k(v)=\emptyset$ as $\Class(v)=\Class(u)$. 
	Otherwise, $F^k(u)\neq \emptyset$, thus $F^k(v)\neq \emptyset$. As $v\in I_{F^k(u)}$ and $v\in I_{F^k(v)}$ we deduce that~$F^k(u)$ and~$F^k(v)$ are both included in $B_v$. As the 
	vertices of a given color in $B_v$ are uniquely determined 
	we deduce $F^k(u)=F^k(v)$. 
	Similarly, we argue that $F^\ell(u)=F^\ell(v)$. It follows that $F^k(u)=F^\ell(u)=u\wedge v$.
	
Hence, if $x\wedge y=u\wedge v$ for $x,y\in X$, then we have
$x\wedge y=F^k(x)=F^k(u)$. 
As $\NC(u)=\NC(v)$, we deduce that for all $x,y\in X$ with 
$x\wedge y=u\wedge v$ we have 
$y\in N_\oplus(F^k(x))$ or for all $x,y\in X$ 
with $x\wedge y=u\wedge v$ we have 
$y\not\in N_\oplus(F^k(x))$. Then it follows from Corollary~\ref{cor:fund}
that at each inner vertex
of $F$ on $X$ we either define a join or a union. Hence, 
$G[X]$ is a cograph with cotree~$F$ restricted to 
$X$ of height at most $r+2$. \hfill
\end{proof}

The function $f(r)$ is most probably far from being optimal. This naturally leads to the following question.

\begin{problem}
	Estimate the growth rate of function $g:\mathbb N\rightarrow\mathbb R$ defined by
\begin{equation}
	g(r)=\sup\,\biggl\{\frac{\chi(G)}{\omega(G)}\mid\lrw(G)\leq r\biggr\}\,.
\end{equation}
\end{problem}
\begin{remark}
One may wonder whether bounding $\chi(G)$ by an affine function of $\omega(G)$ could decrease the coefficient of $\omega(G)$. In other words, is the ratio  $\chi/\omega$ be asymptotically much smaller (as $\omega\rightarrow\infty$) than its supremum?
Note that if $\lrw(G)=r$ and $n\in\mathbb N$, then the graph $G_n$ obtained as the join of $n$ copies of $G$ satisfies $\lrw(G_n)\leq r+1$, $\omega(G_n)=n\omega(G)$ and $\chi(G_n)=n\chi(G)$. Thus
\[
\begin{split}
g(r-1)&\leq \limsup_{\omega\rightarrow\infty}\,\biggl\{\frac{\chi(G)}{\omega(G)}:\lrw(G)\leq r\text{ and }\omega(G)\geq\omega\biggr\}\\
&\leq g(r).	
\end{split}
\]	
\end{remark}

\subsection{Excluding a half-graph.}
\label{sec:HG}
We now apply the assumption that we exclude a 
semi-induced half graph. 
The graph with bounded pathwidth in which we shall 
encode the graph $G$ is the intersection graph~$H$ 
of the intervals $I_u$.
Our remaining task is to determine the adjacency of two vertices 
$u,v$ without reference to the order relation. 

As in the previous section we first consider a partition of the vertex set according to $\Class$- and $\NC$- values.  Precisely, for a $\Class$-value 
 $\sClass$ and a $\NC$-value~$\sNC$ we define 
\[
V_{\sClass,\sNC}=\{v\in V(G):\Class(v)=\sClass\text{ and }\NC(v)=\sNC\}\,.
\]
As one can check from the proof of Theorem~\ref{thm:cog}, 
each of these sets excludes a cograph with cotree of height at 
most $r+2$, which is easily obtained as a transduction of 
the activity interval graph $H$.

We shall thus focus on the edges linking vertices in two different parts  $V_{\sClass_1,\sNC_1}$ and
$V_{\sClass_2,\sNC_2}$. As noted above, for each pair $u,v$ of vertices, the numbers~$\xi(u,v)$ and~$\xi(v,u)$ are easily defined  from $H$, as well as the sets~$F^k(u)$:
By applying Lemma~\ref{lem:or}, we can define (using a first-order transduction) a set $F$ of arcs corresponding to an orientation of a subset of edges of $H$, in which  each vertex $u$ has an incoming arc from each of the vertices $x<u$ with $u\in I_x$ (as these vertices are adjacent to $u$ in $H$). The set $F^k(u)$ is the subset of these in-neighbors (by arcs in $F$) with color in $\Class(u)_k$ (by 
$\Class(u)_k$ we denote the component $\Gamma(F^k(u))$ of 
$\Class(u)$), and $\xi(u,v)$ is the minimum $k$ such that either all the vertices in $F^k(u)$ are in-neighbors of $v$ (by arcs in $F$) or $F^k(u)=\emptyset$.

It follows that we can reduce the problem of defining the adjacencies between $V_{\sClass_1,\sNC_1}$ and $V_{\sClass_2,\sNC_2}$ to the case of a vertex $u\in V_{\sClass_1,\sNC_1}$ and a vertex  $v\in V_{\sClass_2,\sNC_2}$ with 
$\xi(u,v)=k$ and $\xi(v,u)=\ell$.

In order to refine the problem, we first study further the structure of the set $V_{\sClass,\sNC}$. 
	Assume $X,Y\subseteq V(G)$ are distinct subsets with $\max X\leq \max Y$, and  
	assume for contradiction that there exist $u,v \in V_{\sClass,\sNC}$ with $u>v$,  $F^k(u)=X$ and $F^k(v)=Y$.
	Then we have (according to Claim~\ref{claim:tauinc}):
	\begin{align*}
\max X&\leq u<\tau(u)<\tau(X)\,,\\
\max Y&\leq v<\tau(v)<\tau(Y)\,.
	\end{align*}
Then $v\in I_Y$ and because $\max X \leq \max Y \leq v < u < \tau(X) $, also $v\in I_X$. However, 
as $\Class(u)=\Class(v)$ we have
$\Gamma(F^k(u))=\Gamma(F^k(v))$, which contradicts $X\neq Y$ and $I_X\cap I_Y\neq\emptyset$. 
Thus the set 
$V_{\sClass,\sNC}$
is partitioned into sub-intervals corresponding to all distinct values of $F^k(v)$, the intervals being ordered by increasing $\max F^k(v)$. (Note in particular that if $u,v\in V_{\sClass,\sNC}$ and $\max F^k(u)=\max F^k(v)$, then $F^k(u)=F^k(v)$.)

It follows that we can further refine our problem by 
fixing $X=F^k(u)$ and $Y=F^\ell(v)$. Then, if $\xi(u,v)=k$ and $\xi(v,u)=\ell$  we have
\begin{equation}
\label{eq:adj}
\{u,v\}\in E(G)\iff\begin{cases}
	v\in N_\oplus(X)&\text{if }u<v\,,\\
	u\in N_\oplus(Y)&\text{if }u>v\,.
\end{cases}	
\end{equation}

We consider the problem of determining adjacency between two distinct vertices $u$ and $v$ with 
\begin{align*}
\xi(u,v)=\ell_1, &&\Class(u)=\sClass_1, &&\NC(u)=\sNC_1\\
\xi(v,u)=\ell_2, &&\Class(v)=\sClass_2, &&\NC(v)=\sNC_2
\end{align*}

According to \eqref{eq:adj} we have 
\begin{align*}
\{u,v\}&\in E(G)\\
&\iff\begin{cases}
	|\NC(v)\cap \Class(u)_{\ell_1}|&=1\pmod 2\\&\text{if }u<v\,,\\
|\NC(u)\cap \Class(v)_{\ell_2}|&=1\pmod 2\\&\text{if }u>v\,,
\end{cases}	\\
&\iff\begin{cases}
	|\sNC_2\cap (\sClass_1)_{\ell_1}|=1\pmod 2&\text{if }u<v\,,\\
	|\sNC_1\cap (\sClass_2)_{\ell_2}|=1\pmod 2&\text{if }u>v\,.
\end{cases}	
\end{align*}

Of course, if 
$|\sNC_2\cap (\sClass_1)_{\ell_1}|=|\sNC_1\cap (\sClass_2)_{\ell_2}|\pmod 2$ the adjacency of $u$ and $v$ can be computed without knowing whether $u<v$ or $u>v$.
Hence the only difficult case is the case where
$|\sNC_2\cap (\sClass_1)_{\ell_1}|\neq |\sNC_1\cap (\sClass_2)_{\ell_2}|\pmod 2$.

	We define the partial map $\zeta$ on (a subset of the set of all) pairs of vertices $(u,v)$ with $u<v$ defined as follows:
	$\zeta(u,v)=(v',u)$, where $v'$ is maximal with $v'<u$, 
	$\xi(u,v')=\xi(u,v)$, $\xi(v',u)=\xi(v,u)$, $\NC(v')=\NC(v)$, $\Class(v')=\Class(v)$, and  
	$F^{\xi(v,u)}(v')=F^{\xi(v,u)}(v)$.

\begin{lemma}
	Assume $G$ does not contain a semi-induced half-graph of order $p$. Then there exists no 
	sequence $u_1,\dots,u_{p+1},v_1,\dots,v_{p+1}$ with 
$\zeta(u_i,v_i)=(v_{i+1},u_i)$, $\zeta(v_{i+1},u_i)=(u_{i+1},v_{i+1})$, and 
\[
\begin{split}
	&|\NC(v_1)\cap \Class(u_1)_{\xi(u_1,v_1)}|\\
	&\qquad\neq|\NC(u_1)\cap \Class(v_1)_{\xi(v_1,u_1)}|
	\pmod 2.
\end{split}
	\] 
\end{lemma}
\begin{proof}
Whether $u_i$ and $v_j$ are connected in $G$ depends only on their relative order with respect to $<$. Because $u_i < v_j$ if and only if $i\leq j$, if sequence would exist, we would have 
	either $\{u_i,v_j\}\in E(G)\iff i\leq j$ or 
	$\{u_i,v_j\}\in E(G)\iff i> j$ hence either $u_1,\dots,u_p,v_1,\dots,v_p$ or $u_2,\dots,u_{p+1},v_1,\dots,v_p$ semi-induce a half-graph of order $p$.\hfill
\end{proof}

In the following, we assume that $G$ does not contain a semi-induced half-graph of order $p$.
Fix $\ell_1,\sNC_1,\sClass_1,\ell_2,\sNC_2$ and $\sClass_2$ with $|\sNC_2\cap (\sClass_1)_{\ell_1}|\neq$  \mbox{$|\sNC_1\cap (\sClass_2)_{\ell_2}|\pmod 2$}.

For $C\subseteq \IC(u)$ there exists a unique subset $M$ of vertices such that $u\in I_M$ and  
$\Gamma(M)=C$. We denote this subset by $N_C(u)$.
Let
\begin{equation*}
\begin{split}
	Z(u)=\{u'&\in V_{\Class(u),\NC(u)}\mid\\
	& F^k(u')=F^k(u)\text{ and }
	N_C(u')=N_C(u)\}.	
\end{split}
\end{equation*}
Note that \mbox{$u'\in Z(u)\Leftrightarrow u\in Z(u')\Leftrightarrow Z(u)=Z(u')$}.
The sets $Z(u)$ form a partition of $V$ refining the partition formed by the sets $V_{\sClass,\sNC}$.

\begin{lemma}
Let $u\in V_{\sClass_1,\sNC_1}$ and $v,v'\in V_{\sClass_2,\sNC_2}\cap V^{>u}$ be such that
$\xi(u,v)=\xi(u,v')=\ell_1$ and $\xi(v,u)=\xi(v',u)=\ell_2$.
Then $\zeta(u,v)=\zeta(u,v')$.
\end{lemma}
\begin{proof}
	Let $(w,u)=\zeta(u,v)$. By definition $w$ is maximal such that $w<u$, $\xi(u,w)=\ell_1$, $\xi(w,u)=\ell_2$, \mbox{$\NC(w)=\sNC_2$}, $\Class(w)=\sClass_2$, and $F^{\ell_2}(w)=F^{\ell_2}(v)$. However, as $u\in F^{\ell_2}(v)$ and $u\in F^{\ell_2}(v')$ and as these sets have the same $\gamma$-colors we have $F^{\ell_2}(v')=F^{\ell_2}(v)$. Hence $\zeta(u,v')=(w,u)$ as well.\hfill
\end{proof}

This claim actually shows that $Z(u)$ and $Z(v)$ can be partitioned into at most $p+1$ subsets 
(where the partition of $Z(u)$ only depends on $\ell_2,\sClass_2,\sNC_2$ and the partition of $Z(v)$ only depends on $\ell_1,\sClass_1,\sNC_1$). Precisely, $Z(u)_0$ is the set of all $u'\in Z(u)$ such that there exists no $v'\in Z(v)$ with $v'>u'$, $\xi(u',v')=\ell_1$ and $\xi(v',u')=\ell_2$.  $Z(u)_1$ is the set of all $u'\in Z(u)$ such that there exists  no $v'\in Z(v)$ and  $u''\in Z(u)$ with  $u''>v'>u'$,
$\xi(u',v')=\ell_1$ and $\xi(v',u')=\ell_2$,
 and $\zeta(v',u'')=(u',v')$. We inductively define the partitions of $Z(u)$ (and $Z(v)$) by following this pattern.

	We define a marking on $V$ as follows for each $\ell_1,\ell_2,\sClass_1,\sClass_2,\sNC_1,\sNC_2$ 
	with 
	$|\sNC_2\cap (\sClass_1)_{\ell_1}|\neq$ \mbox{$|\sNC_1\cap (\sClass_2)_{\ell_2}|\pmod 2$}, 
	and for each equivalence class $Z(u)$ of $V_{\sClass_1,\sNC_1}$
	we consider the equivalence class $Z(v)$ of $V_{\sClass_2,\sNC_2}$ that contains the elements $v'$ such that \mbox{$\xi(u,v)=\ell_1$} and $\xi_(v,u)=\ell_2$. Taking advantage of the fact that there are at most $2p+1$ alternations we mark the elements of $Z(u)$  by marks
	$M_{\ell_1,\sClass_1,\sNC_1,\ell_2,\sClass_2,\sNC_2,i}$
	and those in  $Z(v)$ by marks $M_{\ell_2,\sClass_2,\sNC_2,\ell_1,\sClass_1,\sNC_2,j}$ in such a way that we can determine 
	the order of the vertices. 

\begin{observation}
We can add marks to all vertices $u,v$ with $|\NC(v)\cap \Class(u)_{\xi(u,v)}|\neq 
|\NC(u)\cap \Class(v)_{\xi(v,u)}|\pmod 2$, with marks 
$M_{\ell_1,\sClass_1,\sNC_1,\ell_2,\sClass_2,\sNC_2,i}(u)$
and $M_{\ell_2,\sClass_2,\sNC_2,\ell_1,\sClass_1,\sNC_1,j}(v)$ 
such that we have 
\[
(u<v)\quad\iff\quad (i<j).
\]
\end{observation}
	
We are ready to state the main result of this section.

\begin{theorem}
  Let $\Cc$ be a class of graphs of bounded linear rankwidth. Then the following are equivalent:
  \begin{enumerate}
  	\item $\Cc$ is stable,
  	\item  $\Cc$ excludes some semi-induced half-graph,
  	\item   $\Cc$ is included in a first-order transduction of a class~$\Dd$ of bounded pathwidth.
  \end{enumerate}
\end{theorem}
\begin{proof}
	If $\Cc$ is a first-order transduction of a class $\Dd$ of bounded pathwidth then $\Cc$ is stable (by \cite{adler2014interpreting}).
	
	If  $\Cc$ is stable then, as mentioned in the introduction, 
	 the class $\Cc$ excludes some semi-induced half-graph.
	 
	If $\Cc$ excludes some semi-induced half-graph then, according to the results obtained in this section, the class $\Cc$ is included in a first-order transduction of a class with bounded pathwidth.\hfill
\end{proof}

\section{Low Embedded Shrubdepth Covers}
\subsection{Embedded Shrubdepth.}
A class $\Cc$ of graphs has bounded shrubdepth if it is a transduction of a class of (rooted) trees with bounded height. Equivalently, a class $\Cc$ has bounded shrubdepth if it is a transduction of a class of trivially perfect graphs with bounded clique number. Indeed, trivially perfect graphs with no clique of order greater than $t$ are exactly the closures of rooted forests with height at most $t$.

On the other side, trivially perfect graphs are special interval graphs, and there is a natural notion of compatibility of an interval graph with a linear order: an interval graph~$G$ is \emph{compatible} with a linear order $<$ on $V(G)$ if there is an interval representation of $G$ in which no two intervals share an endpoint and $<$ is the linear order of the left endpoints of the intervals.

\begin{definition}
	Let $\Cc$ be a class with bounded shrubdepth. An assignment of a linear order $L(G)$ on $V(G)$ to each $G\in\Cc$ is {\em compatible} if there exists a class $\Dd$ of trivially perfect graphs and a transduction $\mathsf T$ such that for each $G\in \Cc$ there exists $H\in\Dd$ with $G\in\mathsf T(H)$ and~$H$ compatible with $L(G)$.
\end{definition}

\begin{definition}
\label{def:esd}
A class $\Cc$ has \emph{bounded embedded shrubdepth} if there exist integers $h,c$ and a class $\mathscr Y$ of  rooted plane forests with height at most $h$, with leaves colored with colors in $[c]$, and with an assignment of a function $f_v:[c]\times [c]\rightarrow\{0,1\}$ to each internal node~$v$ such that
	for every $G\in\Cc$ there is $Y\in\mathscr Y$ and an injection $b$ from~$V(G)$ to the set of leaves of $Y$ with the following property. For all distinct $u,v\in V(G)$ we have $\{u,v\}\in E(G)$ if and only if $b(u)$ is on the left of $b(v)$ at the greatest
	common ancestor $w=b(u)\wedge b(v)$ (in $Y$) of $b(u)$ 
	and $b(v)$ and $f_w(\gamma(b(u)),\gamma(b(v)))=1$ (or the condition obtained by exchanging $u$ and $v$).
\end{definition}

The following alternative point of view may be helpful: 
\begin{lemma}
	A class $\Cc$ has bounded embedded shrubdepth if it can be obtained from a class $\Dd$ of digraphs with bounded shrubdepth and a compatible linear order assignment~$L:\vec{H}\mapsto L(\vec H)$ on $\Dd$, by the following standard quantifier-free interpretation:
\[
\begin{split}
&G\models E(x,y)\quad\iff\\\quad&(\vec H,L(\vec H)) \models ((x<y)\wedge  E(x,y))\,\vee\, ((x>y)\wedge  E(y,x)).
\end{split}
\]
\end{lemma}
\begin{proof}
	Assume $\Cc$ has bounded embedded shrubdepth. Let $h,c,\mathscr Y$ be as in Definition~\ref{def:esd}. Following the definition, to each $G\in\Cc$ we associate a colored tree $Y\in\mathscr Y$ and an injection $b$ from the vertex set of $G$ to the set of leaves of $Y$, with the property that for all distinct $u,v\in V(G)$ we have $\{u,v\}\in E(G)$ if and only if $b(u)$ is on the left of $b(v)$ at the greatest common ancestor $w=b(u)\wedge b(v)$ (in $Y$) and $f_w(\gamma(b(u)),\gamma(b(v)))=1$ (or the condition obtained by exchanging $u$ and $v$). We consider the interpretation of these trees into digraphs, where the vertex set of the interpreted digraph $\vec H$ is the set of leaves of $Y$, and where there is an arc from $u$ to~$v$ if~$u$ and $v$ are distinct and $f(\gamma(b(u)),\gamma(b(v)))=1$. As an interpretation of a class of colored rooted trees with bounded height the obtained digraphs form a class with bounded shrub-depth. If $L(\vec H)$ is the pre-order corresponding to the embedding of $Y$, then it is clearly a compatible linear order, and the quantifier-free interpretation made explicite in the statement of the lemma constructs $G$ from $\vec H$ and $L(\vec H)$.\hfill
	\end{proof}

\subsection{Low treedepth covers over interval graphs.}
\label{sec:ltd_int}
Let $H$ be an interval graph and let $L$ be a linear order compatible with an interval representation of~$H$. Recall that we assume without loss of generality that no two intervals share an endpoint.
We linearly order~$V(H)$ according to the order of the left endpoint of the corresponding intervals, and orient the edges of $H$ accordingly
(from the larger to the smaller endpoint).
A key property of interval graphs ordered as above 
is that the out-neighborhood of every vertex is a transitive 
tournament. This makes it easy to compute  
 $p$-centered colorings of $H$ and deduce a low treedepth cover: 
A \mbox{$p$-centered} coloring of a graph $G$ is a vertex coloring such that, for any (induced) connected subgraph $H$, either some color $c(H)$ appears exactly once in $H$, or~$H$ gets at least $p$ colors. A class of graphs admits $p$-centered colorings for each integer $p$ if and only if the class has low treedepth covers  (see \cite{Taxi_tdepth}).

We consider a modified version of the first-fit coloring algorithm, 
where each vertex receives the first color not present in its 
$p$-th iterated closed out-neighborhood. 
\begin{claim}
	The obtained coloring is $(p+1)$-centered.
\end{claim}
\begin{proof}
Consider a subset $A$ of vertices inducing a connected subgraph of $H$ with at most $p$ colors. 
Let $a,b$ be respectively the maximum and minimum vertex in~$A$. 
We claim that $b$ has a unique color among the colors received
by $A$. 
It is easily checked that there exists a directed path 
$u_k=a,u_{k-1},\dots,u_1=b$. Then $k\leq p$ as all the $u_i$'s 
are colored differently. Let $v\in A$ be distinct from $u_1,\dots,u_k$. Let $1\leq i<k$ be such that 
	$u_i<v<u_{i+1}$. Then $v$ is adjacent to $u_i$ as $u_i$ is adjacent to $u_{i+1}$. It follows that $v$ has a color different from~$u_1$. 	\hfill
\end{proof}
The advantage of the above coloring is that it is naturally compatible with the linear order $<$. 
\begin{claim}
Let $I$ be a set of $p$ colors.
The restriction of $<$ to the subgraph $G_I$ of $G$ induced by the $p$ colors in $I$ is the pre-order of some rooted forest $Y_I$, such that $G_I\subseteq \mathrm{Clos}(Y_I)$.
\end{claim}
\begin{proof}
	This claim follows directly from the following two facts: there is no interlacement in $<$ between different connected components of $G_I$ (because the intervals of the connected component are separated), and the minimum vertex of a connected 
	induced subgraph has a unique color (thus allowing recursion on the argument).\hfill
\end{proof}

\subsection{Low embedded shrubdepth covers of graphs with bounded linear rankwidth.}

Let $\mathcal{U}$ be a low embedded shrubdepth
cover of a graph $G$. We say that $\mathcal{U}$ is 
\emph{compatible} with an order $<$ of $V(G)$ if 
for each $U\in \mathcal{U}$ there is an embedded
shrubdepth decomposition (of bounded height) of 
$G[U]$ such that if $u<v$, then $u$ is at the 
left of $v$ at $u\wedge v$ in the decomposition tree.

The proof of the next lemma is an adaptation of \cite[Section 5.1]{SBE_drops}, in which it is proved that first-order transductions of bounded expansion classes admit low shrub-depth covers. In order to adapt the proofs to our setting, we need to start with a low tree-depth cover equipped with a compatible linear order. As we consider only interval graphs we will be able to use the special low tree-depth colorings defined in Section~\ref{sec:ltd_int}, which have a natural compatible linear order.

\begin{lemma}
Let $\Cc$ be a class of interval graphs with bounded clique-number and let~$\mathsf{T}$ be a transduction (over the vocabulary of graphs).
Then for each $G\in\Cc$ and each linear order $<$ of 
$V(G)$ compatible with the interval representation of $G$ 
there exists a low embedded shrubdepth cover of $G$
that is compatible with $<$. 
\end{lemma}
\begin{proof}
We first rewrite $\mathsf{T}$ to an almost quantifier-free 
transduction using the quantifier elimination technique of 
\cite{dvovrak2013testing} and \cite{Grohe2011}. We also refer to \cite[Lemma 5.5]{SBE_drops}
 for the statement in exactly the present framework. 
Then for each integer $p$, a low shrubdepth cover of~$\mathsf T(\Cc)$ can be deduced from a depth $q$ low treedepth cover 
of~$\Cc$ 
by essentially keeping the 
same decompositions forests. 
It follows that the decomposition forest for the low shrubdepth 
cover of $\mathsf T(\Cc)$ can be chosen such that $<$ corresponds to a pre-order on its vertex set. From this we conclude 
that~$\mathsf T(\Cc)$  has a low embedded shrubdepth decomposition.\hfill
\end{proof}
\begin{lemma}\label{lem:emb-sdc}
	Every class with bounded linear rankwidth has low embedded shrubdepth covers.
\end{lemma}
\begin{proof}
Let $r$ be an integer and let $G$ be a graph of bounded
linear rankwidth at most $r$. We consider our encoding of 
$G$ as an interval graph $H$ as presented in Section~\ref{sec:lrw}. 
In fact, we consider the orientation of $H$ which corresponds 
for all $u,v$ to the creation of an arc from $u$ to $v$ 
if $v\in N_\oplus(F^{\xi(u,v)}(u))$. This orientation can be 
obtained by a first-order transduction using Lemma~\ref{lem:or}. 
For an integer $p$ we consider a low shrubdepth cover 
at depth $p$ such that the linear order $<$ coincides with the 
pre-order on the decomposition forests. We then only keep 
the arcs oriented according to this linear order.\hfill
\end{proof}

As classes with bounded embedded shrubdepth are obviously transductions of linear orders we have:

\begin{theorem}
	A class admits low linear rankwidth decompositions if and only if it it admits low embedded shrubdepth decompositions.
\end{theorem}
\begin{proof}
	One direction is the consequence of Lemma~\ref{lem:emb-sdc}. The opposite direction follows from the fact that classes with bounded embedded shrubdepth have  bounded linear rankwidth.
	\hfill
\end{proof}

\section{Rankwidth}
In this section we prove that the class of all graphs with rankwidth at most $r+1$ is ``vertex-Ramsey'' for the class of all graphs with rankwidth at most $r$, in the following sense.

\begin{theorem}
	For every integers $r,m$ and every graph~$F$ with rankwidth at most $r$ there exists a graph $G=F^{\bullet m}$ with rankwidth $r+1$ with the property that every $m$-coloring of $G$ contained an induced monochromatic copy of $F$.
\end{theorem}

If $G$ and $H$ are graphs, we denote by $G\bullet H$ the
lexicographic product of $G$ and $H$ with vertex set 
$V(G)\times V(H)$ and where two vertices $(u,v)$ and 
$(x,y)$ are adjacent in $G\bullet H$ if and only if 
either $u$ is adjacent with $x$ in $G$ or $u=x$ and 
$v$ is adjacent with $y$ in $H$. 
Note that this operation, though non-commutative, 
is associative. We write $G\otimes K_1$ for graph 
that is obtained from $G$ by adding a new vertex that is 
adjacent to all vertices of $G$.  

A referee pointed out that the following lemma might follow from some known results. Nevertheless, we include its short proof for completeness.

\begin{lemma}
For all graphs $G,H$ (with at least one edge) we have
\[
\mathrm{rw}((G\bullet H)\otimes K_1)=\max(\mathrm{rw}(G\otimes K_1),\mathrm{rw}(H \otimes K_1)).
\]	
\end{lemma}
\begin{proof}
	Let $Y_G$ and $Y_H$ be sub-cubic trees with set of leaves $V(G)\cup\{\alpha\}$ and $V(H)\cup\{\beta\}$, witnessing the rankwidths of $G\otimes K_1$ and $H\otimes K_1$.
	Consider $|G|$ copies of $Y_H$ and glue these copies on $Y_G$ by identifying each leaf of $Y_G$ that is a vertex of $G$ with the vertex $\beta$ of the associated copy. The obtained tree is a rank-decomposition of $(G\bullet H)\otimes K_1$.
	
	For the other inequality, notice that $G\otimes K_1$ and $H\otimes K_1$ are both induced subgraphs of $(G\bullet H)\otimes K_1$.
	
	\hfill
\end{proof}

Actually we can improve the previous lemma by considering substitution instead of lexicographic product. Substitution of $H$ in $G$ at a vertex $v$ is obtained by replacing $v$ by a copy of $H$ with adjacencies of vertices in $H$ being those of $v$. Thus $G\bullet H$ is the substitution of $H$ at every vertex of $G$.

\begin{corollary}
	Closing a class by substitution increases the rankwidth by at most one.
\end{corollary}

For a class $\Cc$, let $\Cc\otimes K_1$ denote the class \mbox{$\{G\otimes K_1\mid G\in\Cc\}$}, and let $\Cc^{\bullet}$ denote the closure of $\Cc$ under lexicographic product.
As a direct consequence of the previous lemma we have
\begin{corollary}
\label{cor:lex}
	For every class of graphs $\Cc$ with bounded rankwidth we have
\begin{equation}
	\mathrm{rw}(\Cc)\leq \mathrm{rw}(\Cc^{\bullet})=\mathrm{rw}(\Cc\otimes K_1)\leq \mathrm{rw}(\Cc)+1.
\end{equation}
\end{corollary}
(Indeed, $G\otimes K_1\subseteq_i G\bullet H$ if $H$ contains at least an edge.)
For instance, as $\mathrm{rw}(P_4\otimes K_1)=2$ we deduce that $\mathrm{rw}(\{P_4\}^{\bullet})=2$.

\begin{theorem}
	For every graph $G$ with rankwidth $r$ and every integer $m$ there exists a graph $G'$ with rankwidth  $\mathrm{rw}(G'\otimes K_1)\leq r+1$, such that for every partition of the vertex set of $G'$ into $m$ classes at least one class contains the vertex set of an induced copy of $G$.
\end{theorem}
\begin{proof}
	We inductively define graphs $G^{\bullet i}$ for $i\geq 1$:
	$G^{\bullet 1}=G$ and, for $i\geq 1$ we let $G^{\bullet (i+1)}=G^{\bullet i}\bullet G$ $=G\bullet G^{\bullet i}$.
	According to Corollary~\ref{cor:lex} we have
	$\mathrm{rw}(\{G^{\bullet i}\mid i\in\mathbb N\})\leq r+1$.
	
	We prove by induction over $m$ that in every $m$-partition of $G'=G^{\bullet m}$ one class induces a subgraph with a copy of $G$. If $m=1$ the result is straightforward.
	Let $m>1$.
	Consider a partition $V_1,\dots,V_m$ of the vertex set of $G^{\bullet m}$. If all the $G^{\bullet (m-1)}$ forming $G^{\bullet m}$ contain a vertex in $V_m$, then $G^{\bullet m}[V_m]$ contains an induced copy of~$G$. Otherwise, there is a copy of $G^{\bullet (m-1)}$ in $G^{\bullet m}$ whose vertex set is covered by $V_1,\dots,V_{m-1}$. By induction hypothesis $G^{\bullet (m-1)}[V_i]$ contains an induced  copy of $G$.
	
	\hfill
\end{proof}

\begin{corollary}
\label{cor:herw}
	Let $\mathcal F$ be a proper hereditary class of graphs. Then there exists a class $\Cc$ with bounded rankwidth such that for every integer $m$ there is $G\in\Cc$ with the property that for every partition of $V(G)$ into $m$ classes, one class induces a graph not in $\mathcal F$
\end{corollary}

\begin{corollary}
\label{cor:cog}
	The class of graphs with rankwidth at most $2$ does not have the property that its graphs can be vertex partitioned into a bounded number of cographs, or circle graphs, etc.
\end{corollary}

\citet{rw_polychi} announced independently that classes with  bounded rankwidth 
are polynomially $\chi$-bounded. We give here a lower bound on the degrees of the involved polynomials.

\begin{theorem}
\label{thm:degrw}
For $r\in\mathbb N$, let $P_r$ be a polynomial such that for every graph $G$ with rankwidth at most $r$ we have 
$\chi(G)\leq P_{r}(\omega(G))$.
Then
$\deg P_r=\Omega(\log r)$.
\end{theorem}
\begin{proof}
		As $\chi(G\bullet H)=\chi(G\bullet K_{\chi(H)})$
		(\cite{geller1975chromatic})  and as 
	$\chi(G\bullet K_{\chi(H)})\geq \chi(H)\chi_f(G)$ we deduce that $\chi(G\bullet H)\geq \chi_f(G)\chi(H)$
	hence
	$\chi(F^{\bullet n})\geq \chi_f(F)^n$.
	Thus $\chi(F_n)\geq \omega(F_n)^{\frac{\log\chi_f(F))}{\log\omega(F)}}$ and we have
\[
\deg P_r\geq\limsup_{\substack{\mathrm{rw}(G)\leq r\\\omega(G)\rightarrow\infty}}\frac{\log \chi(G)}{\log \omega(G)}\geq \sup_{\mathrm{rw}(G\otimes K_1)\leq r}\frac{\log \chi_f(G)}{\log \omega(G)}.
\]

For every sufficiently large integer $n$ there exists a triangle-free graph $G_n$ with $\chi_f(G_n)\geq\frac{1}{9}\sqrt{\frac{n}{\log n}}$ (see~\cite{R3t}). 
As we can choose $n>\mathrm{rw}(G\otimes K_1)$ we deduce that for every sufficiently large integer $r$ we have
\[
\deg P_r\geq \biggl(\frac{1}{2\log 2}-o(1)\biggr)\log r
\]

\hfill
\end{proof}

Apart from solving Conjecture~\ref{conjecture:lrw},  unveiling the surprising structural difference between rankwidth and linear rankwidth (as demonstrated by Theorem~\ref{thm:cog} and Corollary~\ref{cor:cog}) is one of the highlights of this paper.

\section{Conclusion, Open Problems, and Future Work}

In this paper, several aspects of classes with bounded linear-rankwidth have been studied, both from  a (structural) graph theoretical and a model theoretical points of view.

On the one hand, it appeared that graphs with bounded linear rankwidth do not form a ``prime'' class, in the sense that they can be further decomposed/covered using pieces in classes with bounded embedded shrubdepth. As an immediate corollary we obtained that classes with bounded linear rankwidth are linearly $\chi$-bounded.
Of course, the $\chi/\omega$ bound obtained in Theorem~\ref{thm:cog} is most probably very far from being optimal.

On the other hand, considering how graphs with linear rank-width at most $r$ are encoded in a linear order or in a graph with bounded pathwidth with marginal ``quantifier-free'' use of a compatible linear order improved our understanding of this class in the first-order transduction framework. Particularly, this allowed us to confirm Conjecture~\ref{conjecture:lrw}, which we can state as follows:
 If a first-order transduction of linear orders (i.e. a class with bounded linear rank-width) does not allow to define linear orders (i.e. is stable), then there was no need to construct the class from linear orders, and the same class could have been defined as a first-order transduction of a class with bounded pathwidth.

These two aspects merge in the study of classes with low linear rankwidth covers, which generalize structurally bounded expansion classes. On the first problems to be solved on these class, two arise very naturally:

\begin{problem}
\label{pb:lrwT}
Is it true that every first-order transduction of a class with low linear-rankwidth covers has low linear-rankwidth covers?	
\end{problem}

As a stronger form of this problem, one can also wonder whether classes with low linear-rankwidth covers enjoy a form of quantifier elimination, as structurally bounded expansion class do.

\begin{problem}
\label{pb:lrwNIP}
Is it true that every class with low linear-rankwidth covers is mondadically NIP?	
\end{problem}
Note that it is easily checked that a positive answer to Problem~\ref{pb:lrwT} would imply a positive answer to Problem~\ref{pb:lrwNIP}.

\bigskip

Classes with bounded rankwidth seem to be much more complex than expected and no simple extension of the results obtained from classes with bounded linear rankwidth seems to hold. In particular, these classes seem to be ``prime'' in the sense that you cannot even partition the vertex set into a bounded number of parts, each inducing a graph is a simple hereditary class like the class of cographs (see Corollary~\ref{cor:herw}). However, Conjecture~\ref{conjecture:rw} still seems reasonnable to us.

\section{Acknowledgments}

We thank the anonymous reviewers for their 
careful proofreading and comments improving the
presentation. In particular we thank the reviewers
for pointing us to the work of \citet{kwon2014graphs}.

\end{document}